\documentclass[a4paper,10pt]{article}
\usepackage[utf8]{inputenc}
\usepackage{amsmath, amssymb, amsfonts, amsthm, mathtools}
\usepackage{chemfig}
\usepackage[shortlabels]{enumitem}
\usepackage{bbm} 
\usepackage[scale=0.85]{geometry} 
\usepackage{mathrsfs} 
\usepackage{thmtools,thm-restate} 

\usepackage[colorlinks]{hyperref}
\usepackage[numbers,sort, compress]{natbib}

\usepackage{pgf,tikz}
\usetikzlibrary{shapes,arrows}
\newsavebox\ReactionBox
\sbox\ReactionBox{\schemestart
   $A + B$\arrow{->[$\kappa_1$]}[0,.8]$2B$
   \arrow(@c1.south east--.north east){0}[-90,.25]
   $A^\star + B$\arrow{->[$\kappa_2$]}[0,.8]$A^\star + A$
   \schemestop
}
\newsavebox\ReactionBoxinsulator
\sbox\ReactionBoxinsulator{\schemestart
   $A + B$\arrow{->[$0.1$]}[0,.8]$2B$
   \arrow(@c1.south east--.north east){0}[-90,.25]
   $A^\star + B$\arrow{->[$0.1$]}[0,.8]$A^\star + A$
   \schemestop
}
\newsavebox\ReactionBoxA
\sbox\ReactionBoxA{\schemestart
   $0$\arrow{<=>[$u(t)$][$0.01$]}$A^\star$
   \schemestop
}
\newsavebox\ReactionBoxBU
\sbox\ReactionBoxBU{\schemestart
   $\overline{A}^\star + P$\arrow{<=>[$10$][$10$]}$C$
   \schemestop
}
\newsavebox\ReactionBoxB
\sbox\ReactionBoxB{\schemestart
   $A^\star + P$\arrow{<=>[$10$][$10$]}$C$
   \schemestop
}
\newsavebox\ReactionBoxBI
\sbox\ReactionBoxBI{\schemestart
   $A + P$\arrow{<=>[$10$][$10$]}$C$
   \schemestop
}

\usepackage{xcolor}
\usepackage{tcolorbox}
\usepackage{mdframed}
\usepackage{lipsum}

\usepackage[framed]{matlab-prettifier}

\usepackage{rotating}

\renewenvironment{figure}[1][]{
  \begin{originalfigure}[#1]
    \begin{mdframed}[linecolor=black!30,backgroundcolor=black!5]
}{
    \end{mdframed}
  \end{originalfigure}
}

\renewenvironment{sidewaysfigure}[1][]{
  \begin{originalsidewaysfigure}[#1]
    \begin{mdframed}[linecolor=black!30,backgroundcolor=black!5]
}{
    \end{mdframed}
  \end{originalsidewaysfigure}
}

\makeatletter
\newcommand{\DrawLine}[1]{%
  \begin{tikzpicture}
  \path[use as bounding box] (0,0) -- (\linewidth,0.5);
  \draw[color=#1,dashed,dash phase=2pt]
        (0-\kvtcb@leftlower-\kvtcb@boxsep,0.25)--
        (\linewidth+\kvtcb@rightlower+\kvtcb@boxsep,0.25);
  \end{tikzpicture}%
  }
\makeatother

\theoremstyle{plain} 
\newtheorem{theorem}{Theorem}[section]
\newtheorem{corollary}[theorem]{Corollary}
\newtheorem{lemma}[theorem]{Lemma}
\newtheorem{proposition}[theorem]{Proposition}

\theoremstyle{definition}
\newtheorem{definition}{Definition}[section]

\theoremstyle{remark}
\newtheorem{remark}{Remark}[section]
\newtheorem{example}{Example}[section]

\newcommand{\ZZ}{\mathbb{Z}}
\newcommand{\RR}{\mathbb{R}}
\newcommand{\G}{\mathcal{G}}
\newcommand{\Sp}{\mathcal{X}}
\newcommand{\Cx}{\mathcal{C}}
\newcommand{\Rc}{\mathcal{R}}
\newcommand{\St}{\mathcal{S}}
\newcommand{\Cons}{\St^\perp}
\newcommand{\scal}[2]{\langle #1, #2\rangle}

\newcommand{\Sy}{\mathscr{S}}
\newcommand{\ADP}{\text{ADP}}
\newcommand{\ATP}{\text{ATP}}
\newcommand{\I}{\text{I}}

\newcommand{\D}{{\color{blue}\bfseries D}}
\newcommand{\T}{{\color{brown}\bfseries T}}
\newcommand{\Active}{{\color{brown} A^\star}}
\newcommand{\Ph}{{\color{red}\bfseries P}}

\newcommand{\ml}[1]{\lstinline[style=Matlab-bw]!#1!}

\DeclareMathOperator{\spann}{span}
\DeclareMathOperator{\Image}{Im}
\DeclareMathOperator{\supp}{supp}

\title{A hidden integral structure endows Absolute Concentration Robust systems with resilience to dynamical concentration disturbances}
\author{Daniele Cappelletti\footnotemark[1], Ankit Gupta\footnotemark[1] and Mustafa Khammash\footnotemark[1]}
\date{}

\begin{document}

\footnotetext[1]{Department of Biosystems Science and Engineering, ETH Zurich, Mattenstrasse 26 4058 Basel, Switzerland.}

\tikzstyle{block} = [draw=olive!50!white, rectangle, very thick, fill=olive!5!white, text width=15em, text centered, minimum height=3em, rounded corners]
\tikzstyle{input} = [coordinate]
\tikzstyle{output} = [coordinate]

\maketitle

\begin{abstract}
 Biochemical systems that express certain chemical species of interest at the same level at any positive equilibrium are called ``absolute concentration robust'' (ACR). These species behave in a stable, predictable way, in the sense that their expression is robust with respect to sudden changes in the species concentration, regardless the new positive equilibrium reached by the system. Such a property has been proven to be fundamentally important in certain gene regulatory networks and signaling systems. In the present paper, we mathematically prove that a well-known class of ACR systems studied by Shinar and Feinberg in 2010 hides an internal integral structure. This structure confers these systems with a higher degree of robustness that what was previously unknown. In particular, disturbances much more general than sudden changes in the species concentrations can be rejected, and robust perfect adaptation is achieved. Significantly, we show that these properties are maintained when the system is interconnected with other chemical reaction networks. This key feature enables design of insulator devices that are able to buffer the loading effect from downstream systems - a crucial requirement for modular circuit design in synthetic biology.
\end{abstract}

\section{Introduction}

 The network of chemical interactions of a biochemical system of interest can be complex and involve unknown reaction propensities. One of the main goal of reaction network theory consists in deriving dynamical properties from simpler graphical characteristic of the model, and independently on the specific value of kinetic parameters \cite{F:book, toth:reaction}. The results presented in this paper follow this approach.
 
 A qualitative property of great interest is the capability of a certain chemical species to be expressed with the same concentration at any positive steady state, independently on the initial conditions and on how many steady states are present. Namely, assume that the dynamics of the biochemical system are expressed by the $d-$dimensional ordinary differential equation
 $$\frac{d}{dt}x(t)=f(x(t)).$$
 We say that the $i-th$ species is \emph{absolute concentration robust} (ACR), if there exists an \emph{ACR value} $q$ independent of the initial condition $x(0)$ such that, whenever $x(t)$ tends to a positive vector $\overline{x}$, we have $\overline{x}_i=q$. In the typical cases of interest, the positive steady state $\overline{x}$ that is reached will depend on the initial condition $x(0)$, while the entry $\overline{x}_i=q$ does not. As noted in \cite{AEJ:ACR}, the property of absolute concentration robustness alone does not imply stability of the positive steady states: it only ensures that if a positive steady state $\overline{x}$ exists, then the value of the ACR species at $\overline{x}$ is the ACR value.
 
 Under the assumption of stability, absolute concentration robustness provides a reliable, predictive response to environmental changes, since the species of interest reaches the equilibrium level relative to the new environmental setting, regardless the previous conditions. The existence and importance of this robustness property for various gene regulatory networks and signal transduction cascades is explored in many papers, including \cite{BF:robust, BL:robustness, BG:robustness, MG:high, SRA:robustness, ASBL:robustness, PS:Envz_book, SRG:two, SAF:sensitivity}.
 
 In the Control Theory community, and under the assumption of stability, the absolute concentration robustness is known as ``robustness to disturbances in the initial conditions'' \cite{XD:robust,YHSD:robust}. To achieve robustness with respect to some disturbance, the imbalance caused by the disturbance needs to be measured first. To this aim, a quantity of interest is the \emph{integrator}, which is a function $\phi$ of the system variables whose derivative is exactly the imbalance to be eliminated. At steady state, the derivative of $\phi$ is zero and so needs to be the imbalance. In the setting of absolute concentration robustness, one would like to find an integrator whose derivative is the difference between the ACR species and its ACR value. Unfortunately, in general this cannot be done, as shown in Section~\ref{sec:CT} and as discussed in \cite{XD:robust}. 
 
 In the present paper, we systematically study for the first time the connection between ACR systems and integrators. Specifically, our first contribution is related to the existence of a linear combination of chemical species whose derivative is the difference between the ACR species and its ACR value, multiplied by a monomial. Such linear combination of species is called \emph{constrained integrator} (CI), because it behaves similarly to an integrator given that the monomial does not vanish \cite{XD:robust}. We rigorously prove that such a linear CI always exists for a large class of models that strictly includes the ACR systems introduced in \cite{SF:ACR}. This result has some important consequences: first of all, under the assumption of stability, it implies that the expression of ACR species is not only robust to changes in the initial conditions, but also to disturbances that are applied over time. 
 
 An important application in Synthetic Biology concerns the design of \emph{insulators}. A number of biochemical systems are known to express a specific output if given a certain input. The systems can therefore be considered as modules with different functions. In cells, different modules are combined so that more complex responses to external stimuli become possible \cite{HHLM:modular}. In Synthetic Biology, it is desirable to combine different modules to achieve the same level of complexity \cite{PW:second}. However, when connected, the different modules can affect the dynamics of each other and they can lose the desired dynamical properties they had when considered in isolation \cite{DNS:modular}. In a simplified framework, an \emph{upstream} module processes an input, and its output is fed to a second, \emph{downstream} module to be further processed. Since the information is passed in form of molecules, which are then consumed or temporarily sequestrated by the downstream module, the equations governing the upstream module dynamics are perturbed and its functionality can be affected. Such effect is commonly called \emph{loading effect} \cite{MRLDW:load} or \emph{retroactivity} \cite{DNS:modular, PM:retroactivity}, and needs to be minimized. In other words, the upstream module needs to be \emph{insulated} from the loading effect caused by the downstream module. We propose two ways in which the robustness of the systems studied in this paper can be used to this aim. The first solution is to simply design an upstream module which is robust to loading effects, modeled as a persistent disturbance over time. The second solution is to design an extra component, called \emph{insulator}, which transfers the signal from the upstream module to the downstream module while at the same time shielding the dynamics of the upstream module from retroactivity effects.
 
 We will also show how more theoretical results on reaction network models can be obtained as a consequence of our work. In Reaction Network Theory, the study of steady state invariants consitute an interesting topic of research\cite{SF:ACR, KMDDG:invariants, DDG:invariants, F:book}. In \cite{SF:ACR}, it has been proven that certain graphical properties of the network imply the existence of an ACR species, regardless the choice of kinetic parameters. Such sufficient conditions are generalized in the present work while they are maintained simple to check. Moreover, no way to explicitly determine the ACR value was given in \cite{SF:ACR}, and we fill the gap by proposing a fast linear method to calculate it. Furthermore, a substantial effort in the Reaction Network community is devoted to understand under what conditions dynamical properties of single systems can be lifted to larger systems \cite{FCW:node,JS:atoms, BP:inheritance, FW:multistationarity,GHMS:joining}. Our contribution in this sense consist in proving that, under certain conditions, if an ACR system of the class studied in this paper is part of a larger model, the ACR species is still so in the lager system and its ACR value is maintained. Finally, it is worth mentioning that in the present work we consider the possibility of time-dependent rates for the occurrence of chemical transformations. This is more general than what is usually studied in Reaction Network Theory, with the exception of few works explicitly allowing for this scenario \cite{JH:solving, BC:robust, CNP:persistence, CMW:parasite, GMS:geometric, A:time}.
  
\section{Examples of ACR systems}

\subsection{An illustrative example}
Consider two proteins $A$ and $B$, whose interaction is described by
 \begin{equation}\label{eq:toymodel}
 \begin{split}
  \schemestart
   $A+B$\arrow{->[$\kappa_1$]}$2B$
   \arrow(@c1.south east--.north east){0}[-90,.25]
   $B$\arrow{->[$\kappa_2$]}$A$
  \schemestop
  \end{split}
 \end{equation}
 where the positive constants $\kappa_1$ and $\kappa_2$ describe the propensity of a reaction to occur. If enough proteins are present and they are homogeneously spread in space, then a good model for the time evolution of the concentrations of proteins $A$ and $B$ is given by mass-action kinetics. Specifically, the concentrations of $A$ and $B$ at time $t$, denoted by $x_A(t)$ and $x_B(t)$ respectively, are assumed to solve
 \begin{equation}\label{eq:toymodeleq}
 \frac{d}{dt}\begin{pmatrix}
              x_A(t)\\ x_B(t)
             \end{pmatrix}
             =\kappa_1 x_A(t)x_B(t)\begin{pmatrix} -1\\1\end{pmatrix}
             +\kappa_2 x_B(t)\begin{pmatrix} 1\\-1\end{pmatrix}
 \end{equation}
 It is easy to check that the steady states of \eqref{eq:toymodeleq} are given by states $(\overline x_A, \overline x_B)$ such that either $\overline x_B=0$ or $\overline x_A=\kappa_2/\kappa_1$. Hence, $A$ is an ACR species because its expression at any positive steady state is the same. It is common during biochemical experiments to be able to control the inflow rate of some species (say $B$). Some additional chemical species may also be introduced, with the purpose of degrading some of the present components (in this case, species $C$ is introduced to faster degrade species $B$). After these modifications, (\ref{eq:toymodel}) becomes
 \begin{equation}\label{eq:toymodel_modified}
 \begin{split}
  \schemestart
   $A+B$\arrow{->[$\kappa_1$]}$2B$
   \arrow(@c1.south east--.north east){0}[-90,.25]
   $0$\arrow{->[$u_1(t)$]}$B$\arrow{->[$\kappa_2$]}$A$
   \arrow(@c3.south east--.north east){0}[-90,.25]
   $0$\arrow{<=>[$u_2(t)$][$\kappa_3$]}$C$
   \arrow(@c6.south east--.north east){0}[-90,.25]
   $C+B$\arrow{->[$\kappa_4$]}$0$
  \schemestop
  \end{split}
\end{equation}
 Since we still have
 $$\frac{d}{dt}x_A(t)=-\kappa_1x_B(t)\left(x_A(t)-\frac{\kappa_2}{\kappa_1}\right),$$
 it is still true that the value of $x_A(t)$ will converge to $\kappa_2/\kappa_1$, as long as the functions $u_1, u_2$ and the rates $\kappa_3, \kappa_4$ are such that the species $B$ is not removed from the system too fast. In this paper we will prove a general result describing when such kind of robustness to persistent perturbations is present for ACR systems. 

 \subsection{EnvZ-OmpR osmoregulatory system}
 Consider the system described in Figure~\ref{fig:envz}. The model is proposed in \cite{SF:ACR, SMMA:input} as osmoregulatory system in Escherichia Coli. It is in accordance with experimental observations discussed in \cite{PS:Envz_book,SRG:two,BG:robustness}. According to the model, whose schematics is described in Figure~\ref{fig:envz}, the activation rate of the sensor-transmitter protein EnvZ depends on the medium osmolarity. Then, an active form of EnvZ transfers its phosphoryl group to the sensory response protein OmpR, which becomes OmpR-P and promotes the production of the outer membrane porins OmpF and OmpC. Hence, it is important that the concentration of OmpR-P responds in a reliable, predictive way to changes in the medium osmolarity (which the rate constants $\kappa_i$ in the reaction network of Figure~\ref{fig:envz} depend upon), but not on the initial concentration of the different chemical species involved. As a matter of fact, it follows from the results developed in \cite{SF:ACR} that {OmpR-P} is an ACR species. 
\begin{figure*}[h!]
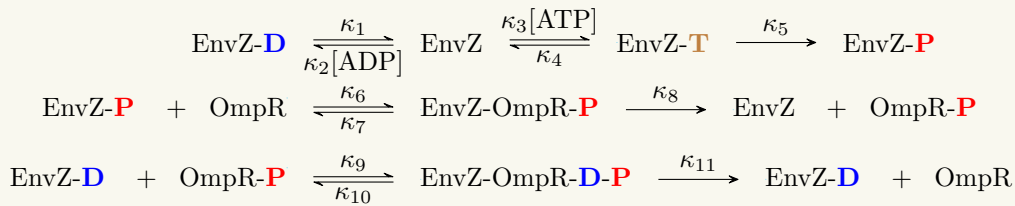

\begin{center}
 \begin{tcolorbox}[width=\textwidth, colback=olive!5!white,colframe=olive!50!white]
$$
\schemestart
EnvZ-\D  \arrow(1--2){<=>[*0{$\kappa_1$}][*0{$\kappa_2[\ADP]$}]} EnvZ \arrow(@2--3){<=>[*0{$\kappa_3[\ATP]$}][*0{$\kappa_4$}]} EnvZ-\T \arrow(@3--4){->[*0{$\kappa_5$}]} EnvZ-\Ph
\arrow(@1.south east--.north east){0}[-90,.35]
EnvZ-\Ph\ \+ OmpR \arrow(5--6){<=>[*0{$\kappa_6$}][*0{$\kappa_7$}]} EnvZ-OmpR-\Ph \arrow(@6--7){->[*0{$\kappa_8$}]} EnvZ \+ OmpR-\Ph
\arrow(@5.south east--.north east){0}[-90,.35]
EnvZ-\D\ \+ OmpR-\Ph \arrow(8--9){<=>[*0{$\kappa_9$}][*0{$\kappa_{10}$}]} EnvZ-OmpR-\D-\Ph \arrow(@9--10){->[*0{$\kappa_{11}$}]} EnvZ-\D\ \+ OmpR
 \schemestop
$$
 \end{tcolorbox}
 \caption{Proposed model for the EnvZ-OmpR signal transduction system in Escherichia Coli, which is able to explain the experimentally observed robustness in the expression of phosphorylated OmpR. In the first line of reactions, EnvZ can bind to ADP and ATP, but only when bound to ATP it can gain a phosphoryl group, and the resulting species is denoted by EnvZ-P. In the second line of reactions, EnvZ-P transfers the phosphoryl group to OmpR, through the formation of an intermediate complex. In the last line of reactions, the phosphoryl group is removed from OmpR-P through the action of EnvZ-D. The concentration of ATP and ADP is assumed to be maintained constant in time.}
 \label{fig:envz}
 \end{center}
\end{figure*}

\section{Necessary terminology and known results}
 
 In order to present the theory we develop, we first need to introduce some terminology. The linear combinations of chemical species appearing on either side of the chemical reactions of interest are called \emph{complexes}, in accordance with the Reaction Network Theory literature. Be aware that the word ``complex'' has usually a different meaning in the Biology literature. We denote by $m$ the number of complexes present in the network, an by $d$ the number of chemical species.
 As an example, the complex of (\ref{eq:toymodel}) are $A+B$, $2B$, $B$, and $A$. Here, $d=2$ and $m=4$. In (\ref{eq:toymodel_modified}) the complexes are $A+B$, $2B$, $0$, $B$, $A$, $C$, and $C+B$, hence $d=3$ and $m=7$. Finally, in the system depicted in Figure~\ref{fig:envz} $d=8$ and $m=10$. Since a complex is a linear combination of species, each complex can be regarded as a vector of length $d$. For example, for the model~(\ref{eq:toymodel}) we can consider $A+B$ as $(1,1)$, $2B$ as $(0,2)$, $B$ as $(0,1)$, and finally $A$ as $(1,0)$. With this in mind, we can define the \emph{stoichiometric subspace} as
 $$\St=\spann_{\RR} \{y_j-y_i\,:\,\text{there is a reaction from }y_i\text{ to } y_j\},$$
 where $y_n$ denotes the $n$th complex, for all $1\leq n\leq m$. For example, for (\ref{eq:toymodel}) we have
 $$\St=\spann_{\RR}\left\{\begin{pmatrix}-1\\1\end{pmatrix}, \begin{pmatrix}1\\-1\end{pmatrix}\right\}=\spann_{\RR}\left\{\begin{pmatrix}-1\\1\end{pmatrix}\right\}.$$
 For (\ref{eq:toymodel_modified}), we have $\St=\RR^3$.
%
 
 In the most general formulation of reaction systems, a (time-dependent) rate function $\lambda_{ij}$ is associated with the reaction from the $i$th to the $j$th complex of the network, and the concentration vector of the different chemical species is assumed to solve the differential equation
 \begin{equation}\label{eq:Gmodel}
  \frac{d}{dt}x(t)=\sum_{1\leq i,j\leq m}(y_j-y_i)\lambda_{ij}(x(t),t),
 \end{equation}
 where if a reaction from the $i$th to the $j$th complex does not exist, then $\lambda_{ij}$ is the zero function. Note that \eqref{eq:Gmodel} simply sums the contributions to the dynamics given by the different chemical reactions. It is also not difficult to show that every solution to \eqref{eq:Gmodel} is necessarily confined within a translation of the stoichiometric subspace. If for all non-zero propensities $\lambda_{ij}$ there exists a positive constant $\kappa_{ij}$ such that
 $$\lambda_{ij}(x(t),t)=\kappa_{ij} \prod_{l=1}^d x_l(t)^{y_{il}},$$
 then the model is a \emph{mass-action system}. In this case, \eqref{eq:Gmodel} can be written as
 $$\frac{d}{dt}x(t)=Y A(\kappa) \Lambda(x(t)),$$
 where $Y$ is a $d\times m$ matrix whose $i$th column is $y_i$, $A(\kappa)$ is a $m\times m$ matrix given by
 $$A(\kappa)_{ij}=\begin{cases}
                 \kappa_{ji}&\text{if }i\neq j\\
                 -\sum_{l=1}^m \kappa_{il} &\text{if }i=j
                \end{cases}$$
 and $\Lambda(x(t))$ is a vector of length $m$ whose $i$th entry is $\prod_{l=1}^d x_l(t)^{y_{il}}$. Examples of mass-action systems are (\ref{eq:toymodel}) and the model in Figure~\ref{fig:envz}.
 
 A directed graph can be associated with a reaction network, where the nodes are given by the complexes and the directed edges are given by the reactions. Such a graph is called \emph{reaction graph}. As an example, (\ref{eq:toymodel}) is a reaction graph, while (\ref{eq:toymodel_modified}) is not because the complex 0 is repeated. The reaction graph corresponding to (\ref{eq:toymodel_modified}) is
 \begin{equation}\label{eq:toymodel_modified_graph}
 \begin{split}
  \schemestart
   $A+B$\arrow(1--2){->[*0{$\kappa_1$}]}$2B$
   \arrow(@1.south east--.north east){0}[-90,.35]
   $C+B$\arrow(7--3){->[*0{$\kappa_4$}]}$0$\arrow(@3--4){->[*0{$u_1(t)$}]}$B$\arrow(@4--5){->[*0{$\kappa_2$}]}$A$
   \arrow(@3--6){<=>[*0{$u_2(t)$}][*0{$\kappa_3$}]}[-90]$C$
  \schemestop
  \end{split}
\end{equation}

 We denote by $\ell$ the number of connected components of the reaction graph associated with the network. For both the networks (\ref{eq:toymodel}) and (\ref{eq:toymodel_modified}) $\ell=2$, as well as for the EnvZ-OmpR osmoregulatory system of Figure~\ref{fig:envz}. Then, we define the \emph{deficiency} of a network as
 $$\delta=m-\ell-\dim \St.$$
 The deficiency of a network has important geometric interpretation, and a collection of classical deficiency theory results can be found in \cite{F:review}. The deficiency of (\ref{eq:toymodel}) is $\delta=4-2-1=1$, and the deficiency of (\ref{eq:toymodel_modified}) is $\delta=7-2-3=2$. Similarly, it can be checked that the deficiency of the EnvZ-OmpR osmoregulatory system in Figure~\ref{fig:envz} is 1.
 
 Finally, we say that a complex $y$ is \emph{terminal} if for all paths in the reaction graph leading from $y$ to another complex $y'$, there is a path leading from $y'$ to $y$. If a complex is not terminal, then it is called \emph{non-terminal}. As an example, the only terminal complexes for (\ref{eq:toymodel}) and (\ref{eq:toymodel_modified}) are $2B$ and $A$.

 We recall that a species is said to be \emph{absolute concentration robust} (ACR) if its concentration at any positive steady state of \eqref{eq:Gmodel} is the same. We are ready to state the following result, as presented in \cite{SF:ACR}.
 \begin{restatable}{theorem}{thmfeinbergweak}\label{thm:feinberg_weak}
Consider a mass-action system, and assume the following holds:
 \begin{enumerate}
  \item there are two non-terminal complexes $y_i$ and $y_j$ such that only one entry of $y_j-y_i$ is non-zero;
  \item the deficiency is 1.
  \item a positive steady state exists.
 \end{enumerate}
 Then, the species relative to the non-zero entry of $y_j-y_i$ is ACR. 
\end{restatable}
Note that a stronger version of Theorem~\ref{thm:feinberg_weak} is proven in \cite{SF:ACR}, which detects steady state invariant that are more general than the equilibrium concentration of a single species. The stronger version is stated in the Supplementary Material as Theorem~\ref{thm:feinberg_strong}, and an extension of it is proven in the present work.

The model (\ref{eq:toymodel}) has deficiency 1, as already observed, has at least one positive equilibrium and the non-terminal complexes $A+B$ and $B$ differ only for the species $A$. Hence, Theorem~\ref{thm:feinberg_weak} applies and $A$ is ACR. It is shown in \cite{SF:ACR} that the EnvZ-OmpR osmoregulatory system in Figure~\ref{fig:envz} also fulfils the hypothesis of Theorem~\ref{thm:feinberg_weak}, with the non-terminal complexes EnvZ-D and EnvZ-D+OmpR-P only differing for the species OmpR-P. As a consequence, OmpR-P is ACR. In Section~\ref{sec:calculate} we will develop a method to explicitly calculate the ACR value through symbolic linear algebra. We note that Theorem~\ref{thm:feinberg_weak} cannot be applied to (\ref{eq:toymodel_modified}) for two reasons: the model is not a mass-action system and its deficiency is 2.

As noted in \cite{AEJ:ACR}, the positive steady states of a system with an ACR species are not necessarily stable. However, as a consequence of the present work (more precisely, as a consequence of Theorem~\ref{thm:control} with $u$ being the zero function), we know the following: if a mass-action system as in Theorem~\ref{thm:feinberg_weak} has an unstable positive steady state, then either the system oscillates around it, or some chemical species is completely consumed, or some chemical species is indefinitely produced. We give here the formal definition of ``oscillation'', as intended in this paper.
\begin{definition}
 We say that a function $g\colon\RR_{\geq0}\to\RR$ \emph{oscillates} around a value $\overline q\in\RR$ if for each $t\in\RR_{\geq0}$ there exist $t_+>t$ and $t_->t$ such that
 $$g(t_+)>q\quad\text{and}\quad g(t_-)<q.$$
\end{definition}
 
 \section{A linear constrained integrator}
 
 \subsection{Control Theory background}\label{sec:CT}
 
 In Control Theory, the focus is usually on systems of differential equations of the form
\begin{equation}\label{eq:system}
  \frac{d}{dt} x(t)=f(x(t), u(t)),
\end{equation}
where $x\colon\RR_{\geq0}\to \RR^{n_x}$ and $u\colon\RR_{\geq0}\to\RR^{n_u}$ for some $n_x,n_u\in\ZZ_{>0}$, and $f$ is a differentiable function. The function $u$ is called the \emph{input of the system}. Further, a quantity of the form $z(t)=a(x(t))$ is of interest, where $a$ is a differentiable function with $a\colon\RR^{n_x}\to\RR^{n_z}$, for some $n_z\in\ZZ_{>0}$. The function $z$ is called the \emph{output} of the system. In the usual setting, one needs to find an appropriate function $u$ such that $z$ is close to a desired level $\overline{z}\in\RR^{n_z}$, either on average or for $t\to\infty$. To this aim, the existence of a function $\phi\colon \RR^{n_x}\to\RR$ such that
$$\frac{d}{dt}\phi(x(t))=z(t)-\overline{z}$$
is of high importance, an is called an \emph{integrator}. The name derives from $\phi(x(t))$ being the integral of the error that needs to be controlled:
$$\phi(x(t))=\phi(x(0))+\int_0^t (z(s)-\overline{z})ds.$$
If the function is fed back to the system and is used to tune the input, then an \emph{integral action} or \emph{integral feedback} is in place \cite{doyle:feedback, astrom:feedback}. One of the main features of an integrator is that the derivative of $\phi(x(t))$ is zero if and only if $z(t)=\overline{z}$. If a function $\tilde\phi\colon \RR^{n_x}\to\RR$ satisfies
$$\frac{d}{dt}\tilde\phi(x(t))=r(x(t))\Big(z(t)-\overline{z}\Big)$$
for some differentiable function $r\colon \RR^{n_x}\to\RR$, then $\tilde\phi$ is called a \emph{constrained integrator} (CI) \cite{XD:robust}. The name derives from the fact that the derivative of $\tilde\phi(x(t))$ is zero if and only if $z(t)=\overline{z}$, provided that $r(x(t))\neq0$. In biology, it is common to find CIs, and the condition $r(x(t))\neq0$ is usually implied by $x(t)\neq 0$ \cite{XD:robust}. Note that in \cite{XD:robust} an explicit distinction between integrators and integral feedbacks is not made.

In the setting of systems with ACR species, the output $z$ can be considered to be the concentration of the ACR species over time, and $\overline{z}$ can be their ACR values. In (\ref{eq:toymodel}), $z(t)=x_A(t)$ and $\overline{z}=\kappa_2/\kappa_1$. A CI (as noted in \cite{XD:robust}) is given by $\tilde\phi(x(t))=x_B(t)$, since
$$\frac{d}{dt}x_B(t)=\kappa_1x_B(t)\left(x_A(t)-\frac{\kappa_2}{\kappa_1}\right).$$
The question of whether an integrator exists can be quickly answered in negative, because any point of the form $(\overline{x}_A,0)$ is a steady state. If an integrator $\phi$ existed, then by choosing $x(0)=(\overline{x}_A,0)$ we would have
$$0=\frac{d}{dt}\phi(x(t))=\overline{x}_A-\frac{\kappa_2}{\kappa_1},$$
which cannot hold expect for a specific value of $\overline{x}_A$. An integrator may still exist in a weaker sense, if we restrict its domain. For example, in this case the function
$\hat\phi(x(t))=\frac{1}{\kappa_1}\log x_B(t)$
would be an integrator, in the sense that if $x_B(t)>0$ then
$$\frac{d}{dt}\hat\phi(x(t))=
\frac{1}{\kappa_1x_B(t)}\left(\kappa_1x_B(t)\left(x_A(t)-\frac{\kappa_2}{\kappa_1}\right)\right)=
x_A(t)-\frac{\kappa_2}{\kappa_1}$$
However, the domain of $\hat\phi$ is not the entire $\RR^2$. Finally, since linear functions could always be extended continuously to the boundaries of $\RR^2_{>0}$, a linear integrator cannot exist for (\ref{eq:toymodel}) even if its domain is restricted.  

\subsection{Existence and characterization}
We state here our result concerning linear CIs. A stronger version is proved in the Supplementary Material. The result is inspired by the analysis carried on in \cite{SF:ACR}, which is here expanded.

For any $n\times l$ real matrix $M$ and real vector $v$ of length $n$, we denote by $(M|v)$ the  $n\times l+1$ matrix obtained by adding the column $v$ at the right of the matrix $A$. Let $1\leq i,j\leq m$. To present our result, we first need to define the preimage
\begin{equation}\label{eq:Gamma}
 \Gamma_{ij}(\kappa)=\left\{\gamma\in\RR^{d+1}\,:\,
\begin{pmatrix}
A(\kappa)^\top Y^\top\,|\, e_i
\end{pmatrix}\gamma
=e_j\right\},
\end{equation}
where $e_n$ denotes the $n$th vector in the canonical basis of $\RR^m$, whose $n$th component is 1 and whose other components are 0. 
The role of $\Gamma_{ij}(\kappa)$ is that of providing vectors $\gamma$ satisfying
\begin{equation}\label{eq:explanation}
\frac{d}{dt}\langle\hat\gamma,x(t)\rangle=\hat\gamma^\top Y A(\kappa)\Lambda(x(t))=\Lambda_j(x(t))-\gamma_i\Lambda_i(x(t)),
\end{equation}
where $\hat\gamma$ is the projection onto the first $d$ components of $\gamma$, and $\langle\cdot,\cdot\rangle$ is the standard scalar product. Under certain assumptions, \eqref{eq:explanation} will provide us with a CI. Then, the projection of $\Gamma_{ij}(\kappa)$ onto the first $d$ coordinates will be of interest, and we will denote it by $\hat\Gamma_{ij}(\kappa)$. The set $\Gamma_{ij}(\kappa)$ can be calculated with symbolic linear algebra. We will also prove in the Supplementary Material how $\hat\Gamma_{ij}(\kappa)$ is connected with $\St^\perp$, which is a set easily described by linear algebra and independent on the rate functions. Specifically, we will prove that if $\xi\in\hat\Gamma_{ij}(\kappa)$, then necessarily
\begin{equation}\label{eq:gamma_S}
 \{\xi+w\,:\,w\in\St^\perp\}\subseteq\hat\Gamma_{ij}(\kappa).
\end{equation}
We will also give sufficient conditions under which the inclusion in \eqref{eq:gamma_S} is an equality.

 As an example, consider the model in Figure~\ref{fig:envz}. Using Matlab, we quickly obtain that a vector $\xi$ is in $\Gamma_{18}(\kappa)$, with
 \begin{equation}\label{eq:calculation_ex}
  \xi_9=\frac{\kappa_1\kappa_3\kappa_5(\kappa_{10}+\kappa_{11})[\ATP]}{\kappa_2(\kappa_4+\kappa_5)\kappa_9\kappa_{11}[\ADP]},
 \end{equation}
and it is shown in the Supplementary Material that
\begin{equation}\label{eq:calculation_ex_hat}
\hat\Gamma_{18}(\kappa)=
\left\{\hat \xi+\begin{pmatrix}w_1\\w_1\\w_1\\w_1\\0\\w_1\\0\\w_1\end{pmatrix}+\begin{pmatrix}0\\0\\0\\0\\w_2\\w_2\\w_2\\w_2\end{pmatrix}\,:\, w_1, w_2\in\RR\right\},
\end{equation}
where $\hat \xi$ is the projection of $\xi$ onto its first $d=8$ coordinates.

The family of models we study in this paper concerns reaction systems with two non-terminal complexes $y_i$ and $y_j$ differing in just one entry, for which $\hat\Gamma_{ij}(\kappa)$ is non-empty. Our first result shows that such a family includes the models studied in \cite{SF:ACR}. The proof can be found in the Supplementary Material.

\begin{restatable}{theorem}{thmfeinbergimpliescond}\label{thm:feinberg_implies_cond}
  Consider a mass-action system, and assume the following holds:
 \begin{enumerate}
  \item there are two non-terminal complexes $y_i$ and $y_j$ such that only one entry of $y_j-y_i$ is non-zero;
  \item the deficiency is 1.
  \item a positive steady state exists.
 \end{enumerate}
 Then, $\hat\Gamma_{ij}(\kappa)$ is non-empty.  
\end{restatable}

As an example, we know already from direct calculation that for the EnvZ-OmpR signaling system \eqref{eq:calculation_ex_hat} holds, which in turn implies that the set $\hat\Gamma_{18}(\kappa)$ is non-empty. However, we could have also derived this information from Theorem~\ref{thm:feinberg_implies_cond}, without explicitly calculating $\hat\Gamma_{18}(\kappa)$.

We note here that the converse of Theorem~\ref{thm:feinberg_implies_cond} does not hold. We show this with an example of a multisite phosphorylation signaling system in the Supplementary Material, which does not fall in the setting of \cite{SF:ACR} but for which we are able to prove absolute concentration robustness regardless the choice of rate constants, as long as a positive steady state exists. Notably, we are also able to derive information on when this occurs without working directly with the differential equation. As a consequence of this example, the family of models we analyze is proven to be strictly larger than that studied in \cite{SF:ACR}. The following holds.

\begin{restatable}{theorem}{thmlcif}\label{thm:lcif}
  Consider a mass-action system. Assume that there are two complexes $y_i$ and $y_j$ only differing in the $n$th entry, and that $\hat\Gamma_{ij}(\kappa)$ is non-empty. Let $\gamma\in\hat\Gamma_{ij}(\kappa)$, and define
 $$q=\gamma_{d+1}^{\frac{1}{(y_j-y_i)_n}}.$$
 Then, either no positive steady state exists or the $n$th species is ACR with ACR value $q$. Moreover, 
 $$\phi(x)=\sum_{i=1}^d \beta_i x_i$$
 is a linear CI with
 $$\frac{d}{dt}\phi(x(t))=\Lambda_i(x(t)) \left(x_n(t)^{(y_j-y_i)_n}-q^{(y_j-y_i)_n}\right)$$
 for any initial condition $x(0)$ if and only if $\beta\in\hat\Gamma_{ij}(\kappa)$.
\end{restatable}

The existence of a linear CI given by Theorem~\ref{thm:lcif} is essential to develop the results presented in the next sections. Before unveiling the consequences of Theorem~\ref{thm:lcif}, however, it is important to stress that a CI does not necessarily constitute a feedback, as one may be tempted to think. Consider 
 \begin{equation*}
 \begin{split}
  \schemestart
   $A+B$\arrow{->[$\kappa_1$]}$A$\arrow{<=>[$\kappa_2$][$\kappa_3$]}$0$
   \arrow(@c1.south east--.north east){0}[-90,.25]
   $B$\arrow{->[$\kappa_4$]}$2B$
  \schemestop
  \end{split}
\end{equation*}
with $\kappa_1\kappa_3=\kappa_2\kappa_4$. It can be shown that the system satisfies the conditions of Theorems~\ref{thm:feinberg_weak} and \ref{thm:feinberg_implies_cond}, with the non-terminal complexes $A+B$ and $A$ differing only in species $A$. Hence, $A$ is ACR and the assumptions of Theorem~\ref{thm:lcif} hold. A linear CI as in Theorem~\ref{thm:lcif} is given by $\phi(x)=-x_B/\kappa_1$, since for this choice
$$\frac{d}{dt}\phi(x(t))=x_B(t)\left(x_A(t)-\frac{\kappa_4}{\kappa_1}\right).$$
However, the quantity $\phi(x(t))$ does not regulate the dynamics of $A$, since
$$\frac{d}{dt}x_A(t)=\kappa_2-\kappa_3x_A(t)$$
does not depend on $x_B(t)$. Since in this case the CI is not acting on the system, it is not surprising that the existence of positive steady states is lost as soon as $\kappa_1\kappa_3\neq\kappa_2\kappa_4$.

It is also worth mentioning that not all systems with ACR species have a linear CI: consider the mass-action system
 \begin{equation*}
 \begin{split}
  \schemestart
   $A+B$\arrow{->[$\kappa_1$]}$B+C$\arrow{<=>[$\kappa_2$][$\kappa_3$]}$2B$
   \arrow(@c1.south east--.north east){0}[-90,.25]
   $B$\arrow{<=>[$\kappa_4$][$\kappa_5$]}$2E$\arrow{->[$\kappa_6$]}$2D$
   \arrow(@c4.south east--.north east){0}[-90,.25]
   $C$\arrow{->[$\kappa_7$]}$A$
   \arrow(@c7.south east--.north east){0}[-90,.25]
   $D$\arrow{->[$\kappa_8$]}$E$
  \schemestop
  \end{split}
\end{equation*}
The model is considered in \cite{AC:ACR}, where it is proven that the species $A$ is ACR. We show in the Supplementary Material that there exists no linear function $\phi$ whose derivative at time $t$ is of the form $r(x(t))(x_A(t)^\gamma-q)$, for some polynomial $r$ and some real numbers $\gamma, q$.  Note that in this case Theorem~\ref{thm:feinberg_implies_cond} does not apply because the deficiency of the network is 2.

 \section{A method to calculate the ACR value}\label{sec:calculate}

 The first interesting consequence of Theorem~\ref{thm:lcif} is that the ACR values of the mass-action systems satisfying the assumption of the theorem can be calculated by finding at least one element of the preimage $\Gamma_{ij}(\kappa)$, and this can be done via a simple symbolic linear algebra calculation. As an example, consider the EnvZ-OmpR osmoregulatory system in Figure~\ref{fig:envz}. Then, Theorem~\ref{thm:lcif} implies that the ACR value of OmpR-P is the value given in \eqref{eq:calculation_ex}. This value is in accordance with the one found in the Supplementary Material of \cite{SF:ACR}, however we found it by calculating a single element in the preimage of a matrix, as opposed to working with the rather complicated differential equation associated with the model. An even more involved examples is dealt with in the Supplementary Material.
 
 \section{Rejection of persistent disturbances}
 
 \subsection{The result}
 
 We state an important consequence of Theorem~\ref{thm:lcif}, a stronger version of which is proven in the Supplementary Material:
 \begin{restatable}{theorem}{thmcontrol}\label{thm:control}
   Consider a mass-action system, with associated differential equation
   $$\frac{d}{dt} x(t)=f(x(t)).$$
   Assume that there are two complexes $y_i$ and $y_j$ only differing in the $n$th entry, and that $\hat\Gamma_{ij}(\kappa)$ is non-empty. Let $q$ be the ACR value of the $n$th species. Consider an arbitrary function $u$ with image in $\RR^d$ such that a solution to
   $$\frac{d}{dt}\tilde x(t)=f(\tilde x(t))+u(\tilde x(t),t)$$
   exists. 
   Assume that there exists a $\hat\gamma\in\hat\Gamma_{ij}(\kappa)$ which is orthogonal to the vector $u(x,t)$ for any $x,t$. Then, for any initial condition $\tilde x(0)$, at least one of the following holds:
   \begin{enumerate}[(a)]
    \item\label{part:infinity_or_zero} the concentration of some species goes to 0 or infinity, along a sequence of times;
    \item\label{part:oscillation_MT} $\tilde x_n(t)$ oscillates around $q$ and $\hat\gamma_k\neq0$ for some $k\neq n$;
    \item\label{part:ACR_convergence} the integral
    $$\int_t^\infty |\tilde x_n(s)-q|ds$$
    tends to $0$, as $t$ goes to infinity.
   \end{enumerate}
 \end{restatable}
 The result implies that if a disturbance orthogonal to a vector $\hat\gamma\in\hat\Gamma_{ij}(\kappa)$ is applied over time, then the stability of the ACR species is maintained: at most, the ACR species can be forced to oscillate around its original ACR value, but it cannot be forced to attain another equilibrium level without causing extinction or overexpression of the chemical species present. We analyze the power of Theorem~\ref{thm:control} by showing some examples of applications.
 
 \begin{example}
  Consider the mass-action system (\ref{eq:toymodel}), which fulfills the assumptions of Theorem~\ref{thm:control} as already observed. Assume the complexes are ordered as $A+B$, $2B$, $B$, and $A$, and the species are ordered alphabetically as $A$, $B$. Hence, the two non-terminal complexes differing in the ACR species $A$ are the 1st and the 3rd, and it is shown in the Supplementary Material that
  \begin{equation}\label{eq:calculation_toy}
   \hat\Gamma_{13}(\kappa)=\left\{\begin{pmatrix}
                                    0 \\ 1
                                   \end{pmatrix}+\begin{pmatrix}
                                    w \\ w
                                   \end{pmatrix}\,:\, w\in\RR\right\}.
  \end{equation}
  Hence, by choosing $w=-1$, we have that 
  \begin{equation}\label{eq:hjhjhjh}
    \begin{pmatrix}
     -1\\0
    \end{pmatrix}\in \hat\Gamma_{13}(\kappa).
  \end{equation}
  This vector is clearly orthogonal to any disturbance acting on the production and degradation rates of the species $B$. Hence, it follows that the stability and the ACR value of the species $A$ is maintained in (\ref{eq:toymodel_modified}), provided that no species is completely removed or indefinitely expressed. Specifically, since the entry of \eqref{eq:hjhjhjh} relative to $B$ is zero, it follows from Theorem~\ref{thm:control} that if all the species concentrations are bounded from below and from above by positive quantities, necessarily the concentration of the species $A$ converges to its ACR value as $t$ goes to infinity, despite the disturbances.
 \end{example}
 \begin{example}[EnvZ-OmpR osmoregulatory system]
  Consider the osmoregulatory system in Figure~\ref{fig:envz}, whose features have already been discussed in the paper. In particular, we know the species OmpR-P is ACR with ACR value \eqref{eq:calculation_ex}. Recall that we ordered the complexes such that the two non-terminal ones differing in OmpR-P are the 1st and the 8th. It follows from \eqref{eq:calculation_ex_hat} that for any chemical species, there is a vector in $\hat\Gamma_{18}(\kappa)$ with the associated entry equal to 0. It follows that even if the production and degradation of any chemical species in the model is tampered with, the stability and the ACR value of the species OmpR-P are maintained, in the sense described by Theorem~\ref{thm:control}. 
  
  We can push the disturbances further. By appropriately choosing $w_1$ and $w_2$ in \eqref{eq:calculation_ex_hat}, we can see that there is vector in $\hat\Gamma_{18}(\kappa)$ whose entries relative to the species OmpR and EnvZ are both 0. Hence, it follows by Theorem~\ref{thm:control} that by tampering with the production and degradation rates of both these species over time, if no extinction and no overexpression occurs, then the concentration of OmpR-P still converges to the value \eqref{eq:calculation_ex}, or oscillates around it.
 \end{example}
 
  As a final remark, we note that \eqref{eq:gamma_S} can be useful in determining whether a vector in $\Gamma_{ij}(\kappa)$ exists, with a specific component equal to 0, say the $n$th one. In fact, the existence of such a vector can be deduced without calculations, if there is a vector $w\in\St^{\perp}$ whose $n$th component is different from 0.
 
 \subsection{Insulating properties}
 
 Here we illustrate how the theory we developed can be utilized to solve the Synthetic Biology problem of retroactivity. As explained in the Introduction, the loading effects caused by a downstream biochemical module can disrupt the functionality of upstream modules, which prevents the implementation of biochemical circuits by interconnecting biochemical modules with different functions \cite{DNS:modular}. A concrete example of loading effect is illustrated in Figure~\ref{fig:application}.
 
 Assume that a mass-action system has two complexes $y_i$ and $y_j$, that are only different in the $n$th component, which corresponds to the species $X$. Assume further that $\hat\Gamma_{ij}(\kappa)$ is non-empty and that a positive steady state exists. Hence, the species $X$ is ACR, with some ACR value $q$. It further follows from Theorem~\ref{thm:control} that, if there exists $\hat\gamma\in\hat\Gamma_{ij}(\kappa)$ with $\hat\gamma_n=0$, then the production and degradation rates of the species $X$ can be arbitrarily perturbed over time by an arbitrary function $u$, without compromising its robustness. Specifically, if the perturbed system is stable and no chemical species is completely consumed, then the concentration of $X$ will still converge to the same ACR value $q$ as in the original mass-action system. The key observation we make here is that the perturbation $u$ can be considered as the loading effect of a downstream module that takes the concentration of species $X$ as input. In this case, the loading effect on the original mass-action is rejected and the concentration of $X$ is maintained at a desired level $q$ at steady state. Further, the concentration of $X$ is maintained approximately constant in the transient dynamics as well, if we assume as done in \cite{DNS:modular} that a separation of dynamics time scale is in place. Specifically, assume
 $$\frac{d}{dt}\tilde x(t)=\frac{1}{\varepsilon}f(\tilde x(t))+ u(\tilde x(t), t)e_n,$$
 for some small $\varepsilon>0$, with $f(\tilde x(t))$ and $u(\tilde x(t), t)$ being of the same order of magnitude. Under the assumption of stability, if $\varepsilon$ is very small then the perturbed system will quickly approach the slow manifold defined by
 $$0=f(\tilde x(t))+ \varepsilon u_t(\tilde x(t))e_n,$$
 where $u_t$ is a function from $\RR^d$ to $\RR^d$ defined by $u_t(x)=u(x,t)$. By Theorem~\ref{thm:control} applied to the disturbance $u_t$, the species $X$ assumes its ACR value at any positive point of the slow manifold, which is exactly what we wanted.
  
 As an example of application, consider the EnvZ-OmpR osmoregulatory system in Figure~\ref{fig:envz}. It follows from \eqref{eq:calculation_ex_hat} that there exists $\hat\gamma\in\Gamma_{18}(\kappa)$ with a zero in the entry corresponding to the ACR species OmpR-P. Hence, the production and degradation rates of OmpR-P can be arbitrarily changed over time, without altering its robustness property, in the sense described by Theorem~\ref{thm:control}. As observed, the statement still holds true if the perturbation is originated by a downstream module that acts on OmpR-P. Hence, the EnvZ-OmpR osmoregulatory system can be used to maintain the expression of OmpR-P at a desired level, which depends on the input rate constants, even if the species OmpR-P is used by a downstream module. Moreover, if the downstream module acts on a slower time scale, the concentration of the species OmpR-P is approximately maintained at the target level at any time point. In Figure~\ref{fig:block}, a diagram describing this situation is proposed.
 
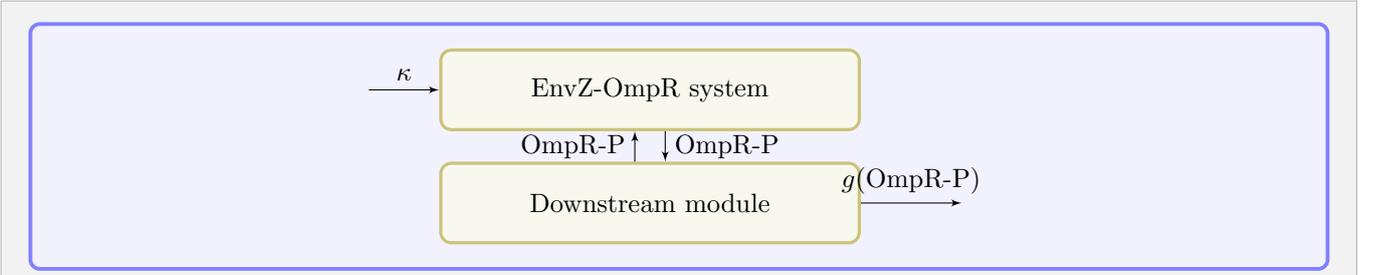
\begin{figure}[h!]
\begin{center}
 \begin{tcolorbox}[colback=blue!5!white,colframe=blue!50!white]
 \begin{center}
\begin{tikzpicture}[auto, node distance=2cm,>=latex']
    \node [input, name=input] {};
    \node [block, right of=input, node distance=3.7cm] (signaling) {EnvZ-OmpR system};
    \node [block, below of=signaling, node distance=1.5cm] (plant) {Downstream module};
    \node [output, right of=plant, node distance=4.1cm] (output) {};
    \draw [draw,->] (input) -- node {$\kappa$} (signaling);
    \draw [->] ([xshift=0.2cm]signaling.south) -- node {OmpR-P}([xshift=0.2cm]plant.north);
    \draw [->] ([xshift=-0.2cm]plant.north) -- node {OmpR-P} ([xshift=-0.2cm]signaling.south);
    \draw [->] (plant) -- node {$g($OmpR-P$)$} (output);
\end{tikzpicture}
\end{center}
 \end{tcolorbox}
 \caption{Proposed use of the EnvZ-OmpR signal transduction system of Figure~\ref{fig:envz} as a controller of a downstream module utilizing OmpR-P. The concentration of OmpR-P is regulated by the EnvZ-OmpR signaling system, with equilibrium given by \eqref{eq:calculation_ex}. The equilibrium can be adjusted by modifying the parameters $\kappa_{ij}$ of the EnvZ-OmpR signaling system (which depend on the medium osmolarity) or the ratio between ADP and ATP present. The output of the downstream module is a function $g$ of the concentration of OmpR-P, which is received as input.}
 \label{fig:block}
 \end{center}
\end{figure}
 
\begin{figure}[h!]
  \begin{center}
    \begin{tcolorbox}[colback=blue!5!white,colframe=blue!50!white]
\begin{center}
\begin{tikzpicture}[auto, node distance=1cm,>=latex']
    \node [input, name=input] {};
    \node [block, right of=input, node distance=4.1cm] (upstream) {Upstream module};
    \node [input, below of=input, node distance=2.5cm, name=input2] {};
    \node [block, below of=upstream, node distance=2.5cm, align=center] (insulator) {Insulator \\[0.3cm] \usebox\ReactionBox};
    \node [block, below of=insulator, node distance=2.5cm] (downstream) {Downstream module};
    \node [output, right of=downstream, node distance=3.7cm] (output) {};
    \draw [draw,->] (input) -- node {$u(t)$} (upstream);
    \draw [->] (upstream) -- node[text width=3cm] {$A^\star$ (not changed by the insulator)}(insulator);
    \draw [draw,->] (input2) -- node {$\kappa_1$, $\kappa_2$, $x_B(0)$} (insulator);
    \draw [->] ([xshift=0.2cm]insulator.south) -- node {$A\approx\displaystyle\frac{\kappa_2 A^\star}{\kappa_1}$} ([xshift=0.2cm]downstream.north);
	\draw [->] ([xshift=-0.2cm]downstream.north) -- node[text width=3.5cm] {$A$ (changed by downstream module)} ([xshift=-0.2cm]insulator.south);
    \draw [->] (downstream) -- node {$g(A)$} (output);
\end{tikzpicture}
\end{center} 
    \end{tcolorbox}
 \caption{
 The upstream module expresses the chemical species $A^\star$ as output. The insulator transfers a multiple of the signal from the upstream module to the downstream module, which is modified to accept as input the concentration of $A$ rather than the concentration of $A^\star$.}
 \label{fig:insulation}
  \end{center}
 \end{figure}
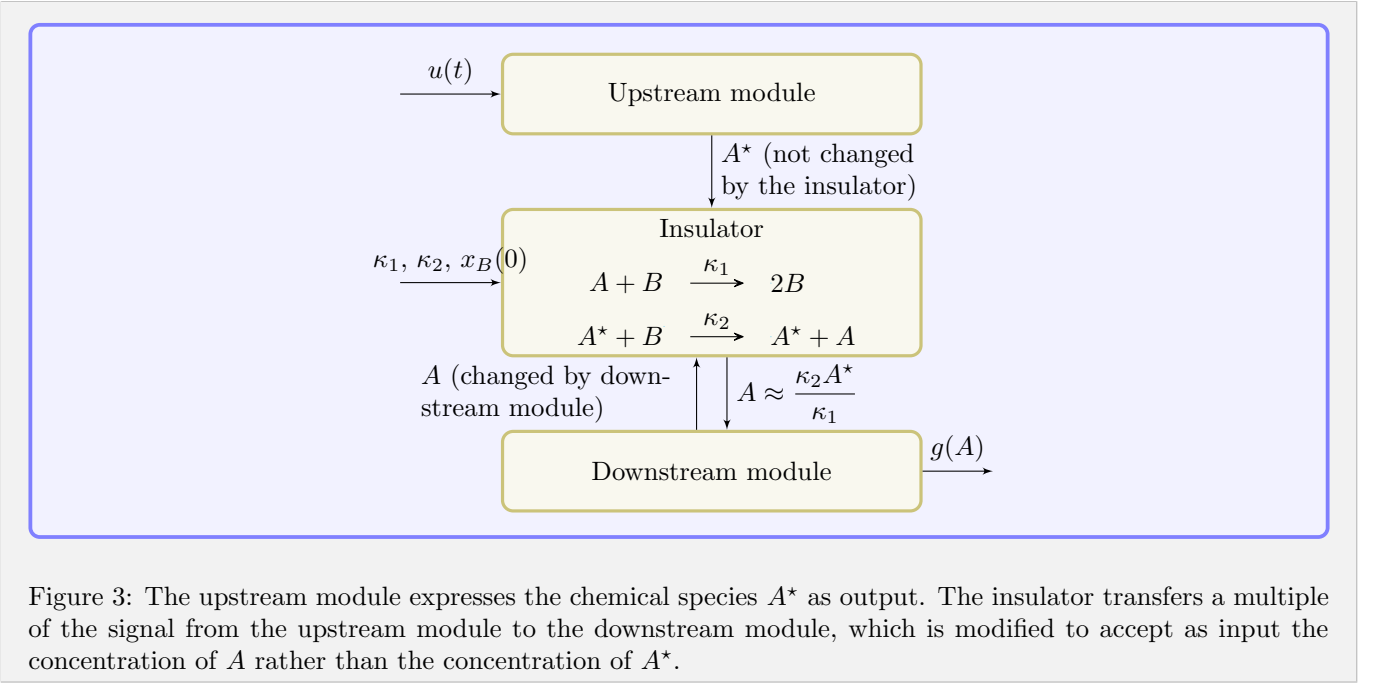
 
 Consider now the case were an upstream module is affected from loading effects. We show how the theory developed in this paper can be used to design an insulator. Assume that the upstream module accepts $u(t)$ as input, and modulates the concentration of the species $A^\star$ accordingly. The species $A^\star$ is then used by a downstream module, which returns a function of the concentration of $A^\star$ as output. The action of the downstream module on the species $A^\star$ causes a loading effect on the upstream module. This loading effect should be reduced. To this aim, we proposed to modify the downstream module such that it acts on a species $A$ rather than on the species $A^\star$, and to include in the system the following module, where $B$ is a species that is not used by neither the upstream nor the downstream module:
 \begin{equation}\label{eq:insulator}
  \begin{split}
   \schemestart
   $A + B$\arrow{->[$\kappa_1$]}[0,.8]$2B$
   \arrow(@c1.south east--.north east){0}[-90,.25]
   $A^\star + B$\arrow{->[$\kappa_2$]}[0,.8]$A^\star + A$
   \schemestop
  \end{split}
 \end{equation}
 Assume stability is reached and that the species $B$ is not completely consumed. Then, at steady state the concentration level of $A^\star$ is fixed, and the concentration of the ACR species $A$ will converge to its ACR value $\kappa_2 x_{A^\star}/\kappa_1$ regardless any disturbance applied to the production and degradation rate of $A$. In fact, a linear CI as in the statement of Theorem~\ref{thm:lcif} is given by $\phi(x)=x_B/\kappa_1$, and at any time point 
 \begin{equation}\label{eq:insulator_CI}
  \frac{d}{dt}\phi(x(t))=x_B(t)\left(x_A(t)-\frac{\kappa_2}{\kappa_1}x_{A^\star}(t)\right).
 \end{equation}
 We further note that if the dynamics of (\ref{eq:insulator}) occur on a faster time scale than the rest of the system, then a slow manifold is quickly approached were the concentration of the species $A$ is maintained at the level $\kappa_2 x_{A^\star}(t)/\kappa_1$ at any time point. In this case, the module (\ref{eq:insulator}) approximately outputs a multiple of the concentration of $A^\star$ over the whole time line. The multiplicative constant can be tuned through the parameters $\kappa_1$ and $\kappa_2$, as well as the time scale that (\ref{eq:insulator}) operates in. The time scale can be further tuned via the concentration of $B$, as it also follows from \eqref{eq:insulator_CI}. In conclusion, the downstream module receives as input a good approximation of a multiple of the concentration of $A^\star$, and its activity does not affect the upstream module, nor (\ref{eq:insulator}). Moreover, (\ref{eq:insulator}) does not affect the upstream module at all, since the species $A$ appears in (\ref{eq:insulator}) as a catalyst and is not changed in the catalysed reaction. The proposed insulating strategy is illustrated in Figure~\ref{fig:insulation}, and it is applied to an example discussed in \cite{DNS:modular} in Figure~\ref{fig:application}.
\begin{figure*}[h!]
  \begin{center}
    \begin{tcolorbox}[width=\textwidth, colback=blue!5!white,colframe=blue!50!white]
 \begin{minipage}[c]{\textwidth}
  \begin{minipage}[c]{0.45\textwidth}
   \begin{center}
   Isolated systems (with $x_{\overline{A}^\star}(t)=x_{A^\star}(t)$)
   
   \vspace*{2em}
    \begin{tikzpicture}[auto, node distance=1cm,>=latex']
     \node [input] (inputA) {};
     \node [block, right of=inputA, node distance=3.5cm, align=center, text width=3.5cm] (upstream) {\usebox\ReactionBoxA};
     \node [output, right of=upstream, node distance=3.5cm] (outputA) {};
     \node [input, below of=inputA, node distance=1.5cm] (inputB) {};
     \node [block, right of=inputB, node distance=3.5cm, align=center, text width=3.5cm] (downstream) {\usebox\ReactionBoxBU};
     \node [output, right of=downstream, node distance=3.5cm] (outputB) {};
     \draw [->] (inputA) -- node {$u(t)$} (upstream);
     \draw [->] (upstream) -- node {$x_{A^\star}(t)$} (outputA);
     \draw [->] (inputB) -- node {$x_{\overline{A}^\star}(t)$} (downstream);
     \draw [->] (downstream) -- node {$x_C(t)$} (outputB);
    \end{tikzpicture}
   \end{center}
  \end{minipage}\hfill
  \begin{minipage}[c]{0.5\textwidth}
   \includegraphics[width=\textwidth]{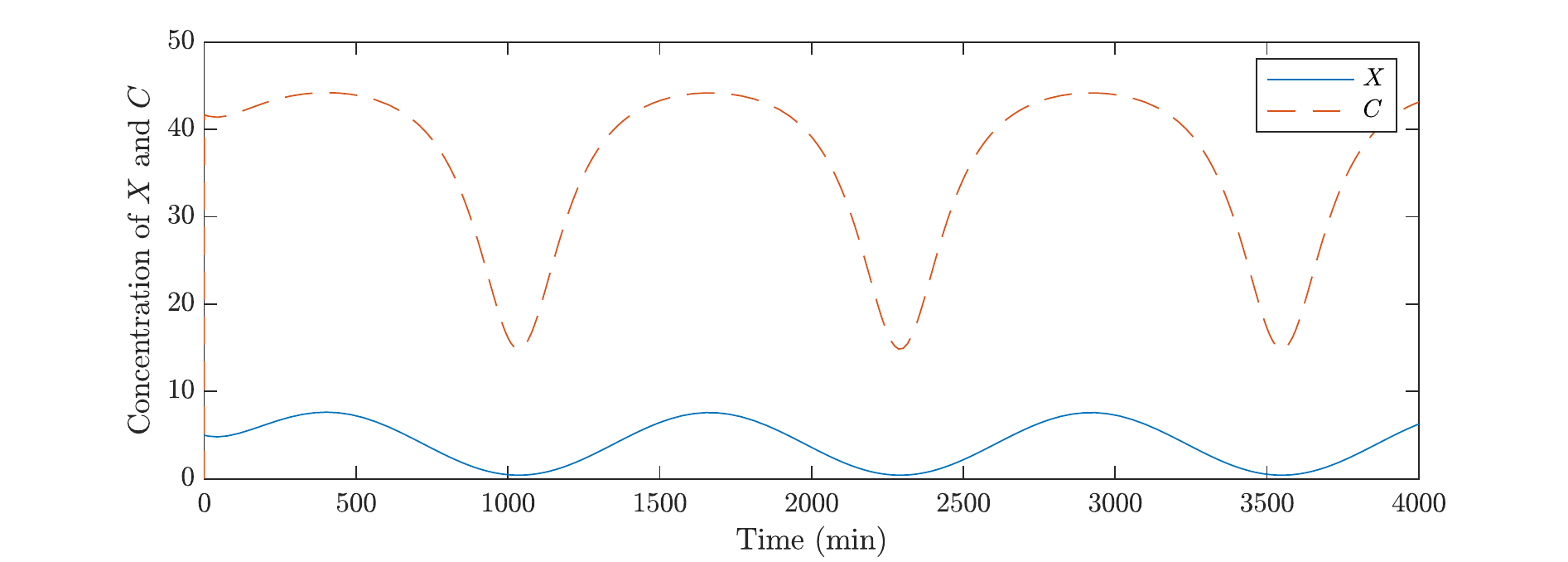}
  \end{minipage}
 \end{minipage}
 \DrawLine{blue!50!white}
 \begin{minipage}[c]{\textwidth}
  \begin{minipage}[c]{0.45\textwidth}
   \begin{center}
  Coupled systems
  
  \vspace*{2em}
   \begin{tikzpicture}[auto, node distance=1cm,>=latex']
     \node [input] (inputA) {};
     \node [block, right of=inputA, node distance=3.5cm, align=center, text width=3.5cm] (upstream) {\usebox\ReactionBoxA};
     \node [block, below of=upstream, node distance=2cm, align=center, text width=3.5cm] (downstream) {\usebox\ReactionBoxB};
     \node [output, right of=downstream, node distance=3.5cm] (outputB) {};
     \draw [->] (inputA) -- node {$u(t)$} (upstream);
     \draw [->] ([xshift=0.2cm]upstream.south) -- node {$x_{A^\star}(t)$} ([xshift=0.2cm]downstream.north);
     \draw [->] ([xshift=-0.2]downstream.north) -- node {} ([xshift=-0.2]upstream.south);
     \draw [->] (downstream) -- node {$x_C(t)$} (outputB);
    \end{tikzpicture}
   \end{center}
  \end{minipage}\hfill
  \begin{minipage}[c]{0.5\textwidth}
   \includegraphics[width=\textwidth]{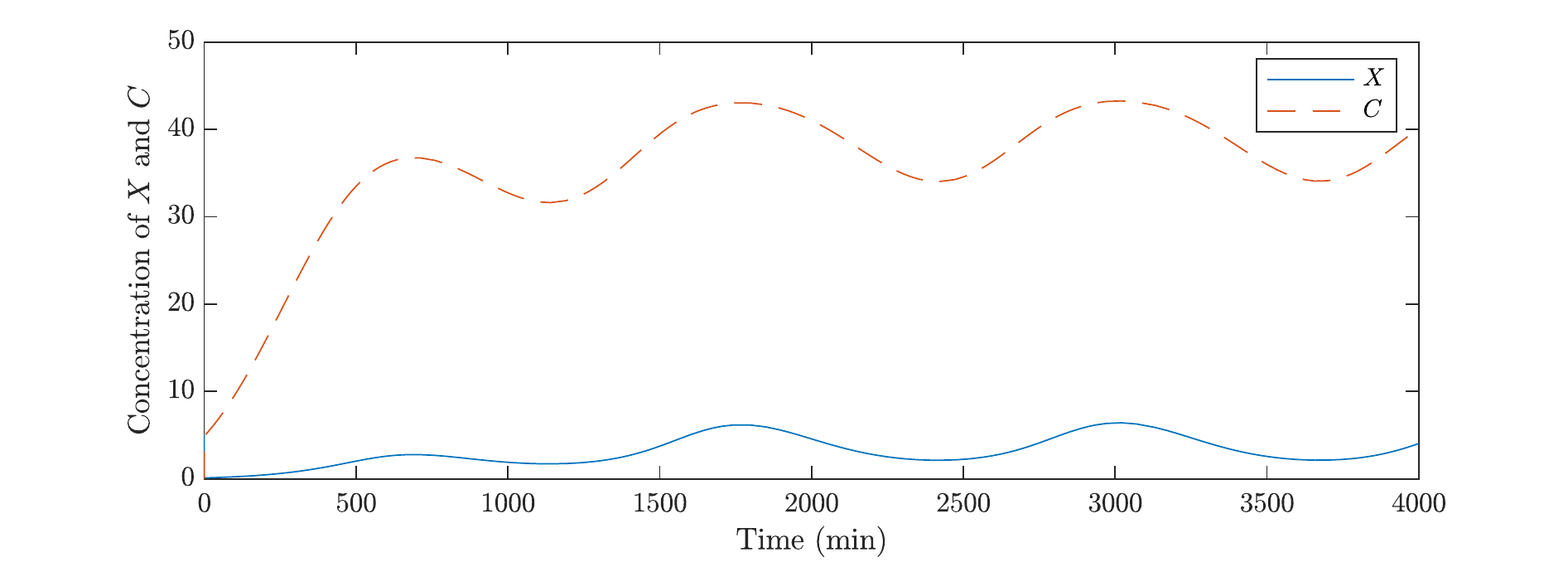}
   
   \includegraphics[width=\textwidth]{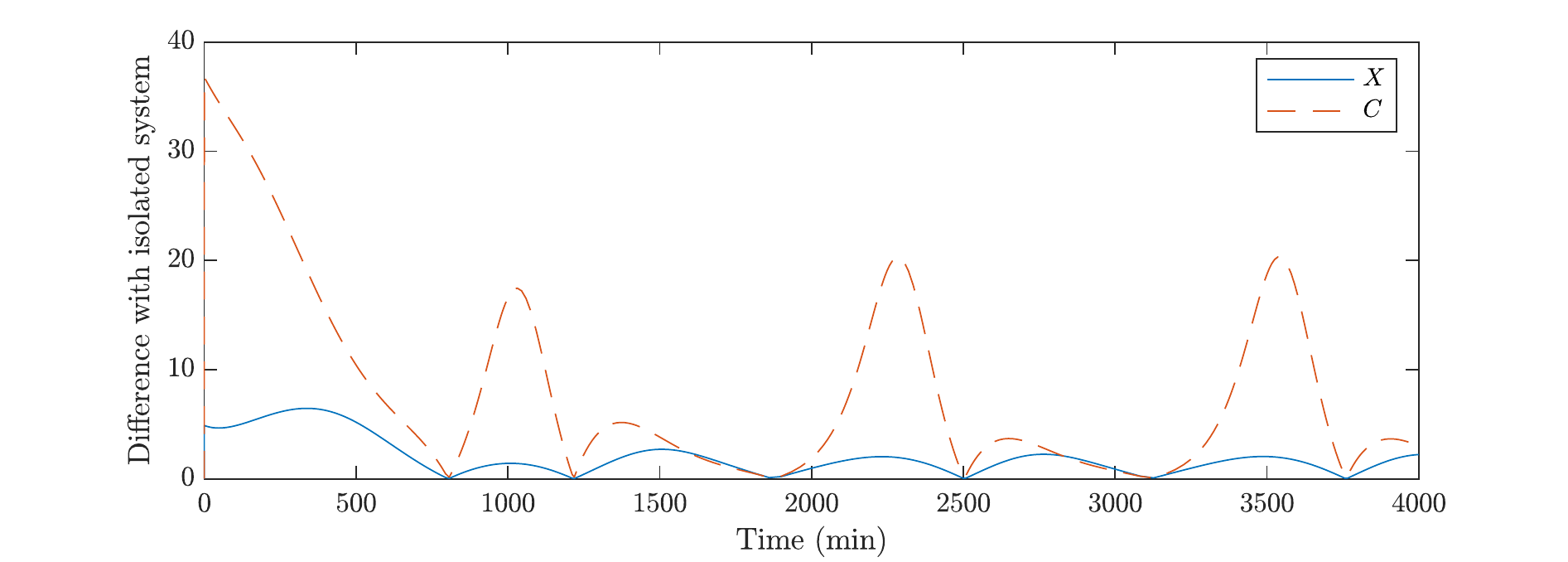}
  \end{minipage}
 \end{minipage}
  \DrawLine{blue!50!white}
 \begin{minipage}[c]{\textwidth}
  \begin{minipage}[c]{0.45\textwidth}
   \begin{center}
  Insulated system
  
  \vspace*{2em}
    \begin{tikzpicture}[auto, node distance=1cm,>=latex']
     \node [input] (inputA) {};
     \node [block, right of=inputA, node distance=3.5cm, align=center, text width=3.5cm] (upstream) {\usebox\ReactionBoxA};
     \node [block, below of=upstream, node distance=2cm, align=center, text width=4cm] (insulator) {\usebox\ReactionBoxinsulator};
     \node [block, below of=insulator, node distance=2cm, align=center, text width=3.5cm] (downstream) {\usebox\ReactionBoxBI};
     \node [output, right of=downstream, node distance=3.5cm] (outputB) {};
     \draw [->] (inputA) -- node {$u(t)$} (upstream);
     \draw [->] (upstream) -- node {$x_{A^\star}(t)$} (insulator);
     \draw [->] ([xshift=0.2cm]insulator.south) -- node {$x_{A}(t)$} ([xshift=0.2cm]downstream.north);
     \draw [->] ([xshift=-0.2]downstream.north) -- node {} ([xshift=-0.2]insulator.south);
     \draw [->] (downstream) -- node {$x_C(t)$} (outputB);
    \end{tikzpicture}
   \end{center}
  \end{minipage}\hfill
  \begin{minipage}[c]{0.5\textwidth}
   \includegraphics[width=\textwidth]{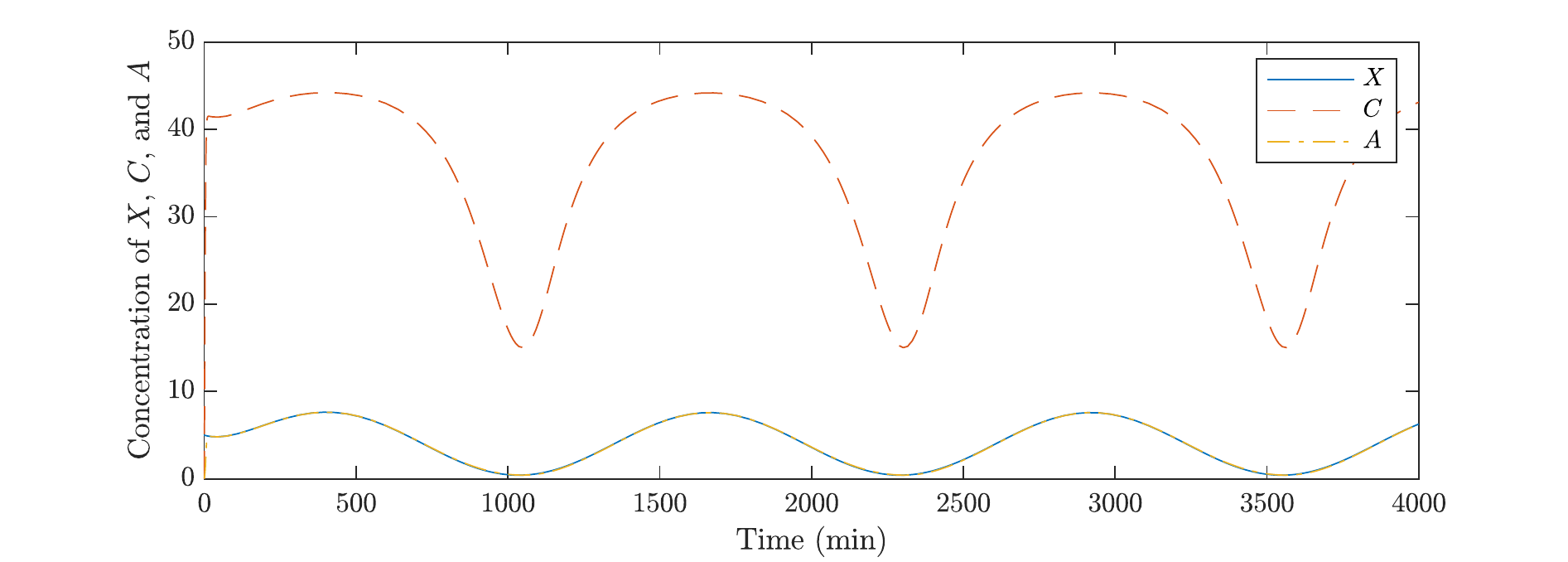}
   
   \includegraphics[width=\textwidth]{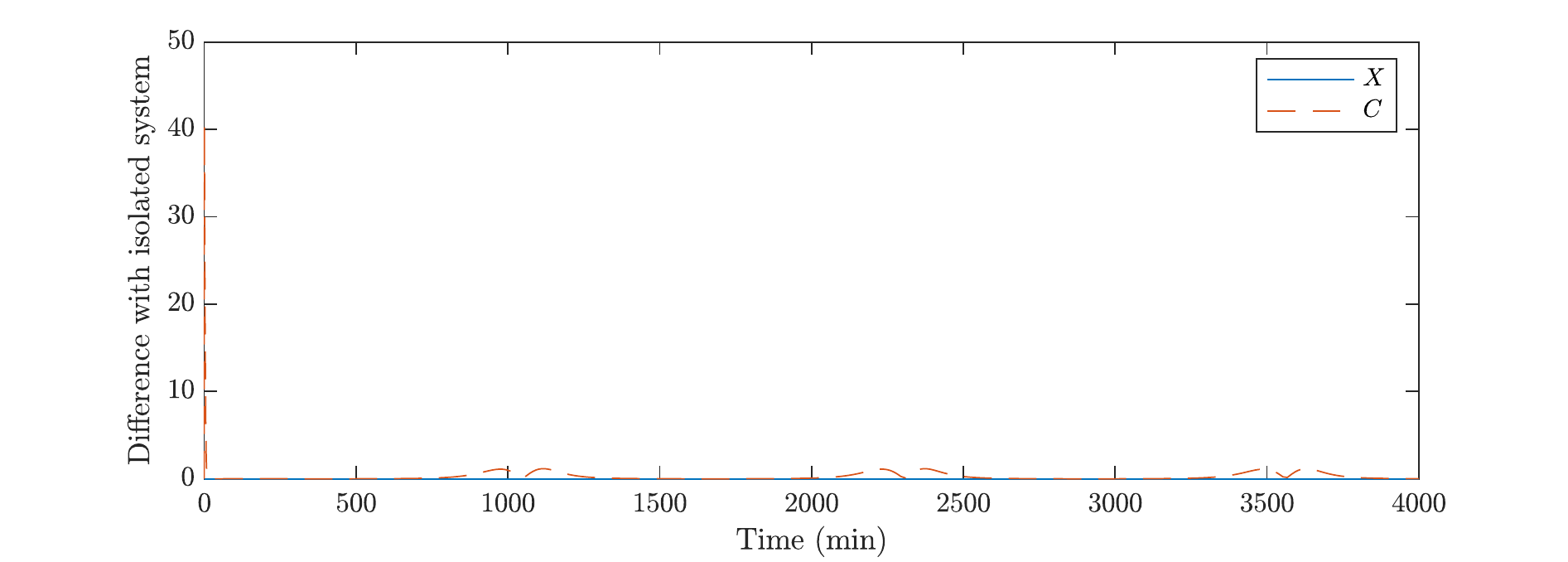}
  \end{minipage}
 \end{minipage}
    \end{tcolorbox}
 \caption{In the example, we let $u(t)=0.04(1+\sin(0.005t))$ and $x_{A^\star}(0)=5$. All the rates are in $1/min$. The ODE solution is calculated in Matlab with \texttt{ode23s}. In the first panel, the two systems are considered in isolation, with the assumption that the concentration of the species $\overline{A}^\star$ is maintained at the same level as the concentration of species $A^\star$, at any time point. In the second panel, the two systems are linked together, and the species $A^\star$ is directly used by the downstream system. The dynamics are completely disrupted by the loading effects, a plot of the absolute value of the difference of the two solutions over time is proposed. In the third panel, the insulator of Figure~\ref{fig:insulation} is utilized, with $x_{A}(0)=0$ and $x_B(0)=20$. The loading effect are practically removed, despite the choice of low rate constants $0.1$. Specifically, the difference of the concentration of $C$ between the solution of insulated system and the solution of the isolated systems spikes quickly to 40, but it decreases to less than 0.5 within 10 minutes, after which is maintained low as illustrated in the second plot of the third panel.}
 \label{fig:application}
  \end{center}
 \end{figure*}
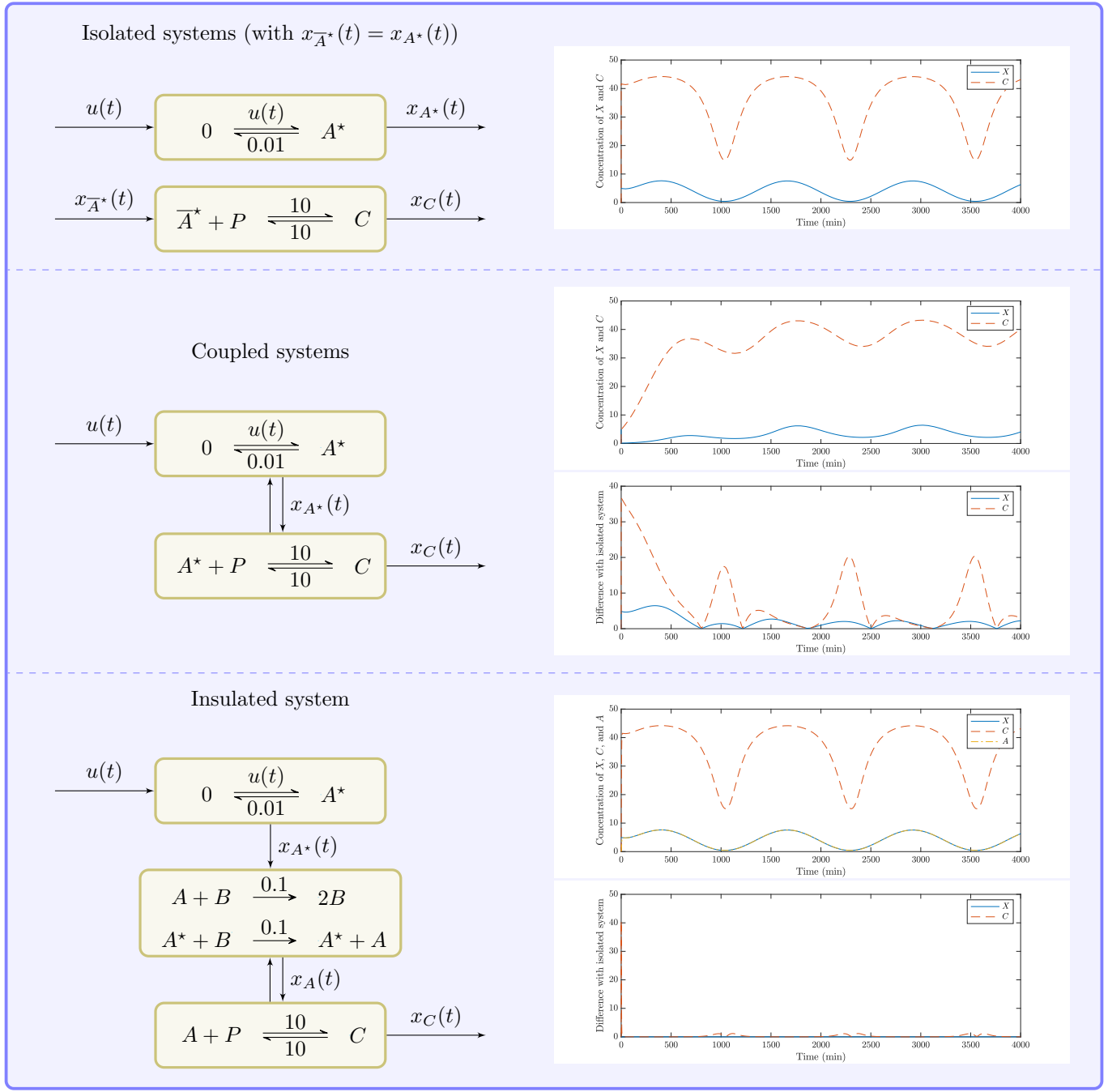
 
 \subsection{Inclusion in larger systems}\label{sec:large}
 
 In the previous section, we have seen how the stability properties of an ACR system can be transferred, in a sense, to a larger model including further chemical transformations and external inputs. As a particular case, Here, we explicitly state when the property of absolute concentration robustness can be lifted from portions of the biochemical system to the whole system. As a consequence, we further extend the set of sufficient conditions of Theorem~\ref{thm:feinberg_weak} that imply the existence of an ACR species. Before stating the relevant result, which is a consequence of Theorem~\ref{thm:lcif}, we need a definition. Given a reaction system $\Sy$, we say that $\Sy'$ is a sub-system of $\Sy$ if it can be obtained from $\Sy$ by canceling some reactions, and if the choice of rate functions for the remaining reactions is maintained. Moreover, if $\Sy$ and $\Sy'$ have $d$ and $d'$ species, respectively, we let $\pi:\RR^d\to\RR^{d'}$ be the projection onto the species of $\Sy'$. The following holds:
 \begin{restatable}{corollary}{thmextension}\label{thm:extension}
  Consider a reaction system $\Sy$, and let $\Sy'$ be a sub-system. Assume that $\Sy'$ is a mass-action system with two complexes $\pi(y_i)$ and $\pi(y_j)$ only differing in the entry relative to the species $X$, and for which $\hat\Gamma_{ij}(\kappa)$ is non-empty. Moreover, assume there exists $\hat\gamma\in\hat\Gamma_{ij}(\kappa)$ such that $\pi(y_l-y_k)$ is orthogonal to $\hat\gamma$ for all $y_k\to y_l$ that are reactions of $\Sy$ but not reactions of $\Sy'$.
  Then, the $X$ is ACR for both $\Sy$ and $\Sy'$, with the same ACR value.
 \end{restatable}
 The proof of a stronger result is in the Supplementary Material. Here we illustrate how the corollary can be applied in the case of EnvZ-OmpR signaling system: consider the reaction system $\Sy$ described in Figure~\ref{fig:larger}, which includes the EnvZ-OmpR osmoregulation system described in Figure~\ref{fig:envz} as a sub-system. We assume that a protein can misfold when the phosphoryl group is transferred from EnvZ to OmpR. Such misfold can be corrected by chaperons, which are independently produced and degraded through a mechanism which we assume unkown, but that does not involve EnvZ or OmpR proteins. We also allow for an arbitrary and persistent external control on the expression level of EnvZ sensor-transmitter protein. Finally, we consider the utilization of OmpR-P as transcription regulatory protein of the outer membrane porins OmpF and OmpC. For our purposes, we assume the details of the transcription mechanism are not known, but that only the protein Ompr-P is involved in the process. As previously done, let the complexes of the EnvZ-OmpR osmoregulation system be ordered from left to right and from top to bottom, such that EnvZ-D and EnvZ-D+OmpR-P are the 1st and the 8th complex, respectively. Also, let the species be ordered according to their appearance from left to right and from top to bottom, as EnvZ-D, EnvZ, EnvZ-T, EnvZ-P, OmpR, EnvZ-OmpR-P, OmpR-P, EnvZ-OmpR-D-P. In particular, Envz is the second species, EnvZ-OmpR-P is the sixth species, and OmprR-P is the seventh species. It follows from \eqref{eq:calculation_ex_hat} that, by choosing $w_1=-\xi_2$ and $w_2=-\xi_7=0$, a vector $\hat\zeta$ is in $\hat\Gamma_{18}(\kappa)$ with:
 \begin{enumerate}
  \item\label{1} $\hat\zeta_2=\hat\xi_2+w_1=0$;
  \item\label{2} $\hat\zeta_6-\hat\zeta_2=\hat\xi_6+w_1+w_2-\hat\xi_2-w_1=\hat\xi_6-\hat\xi_2=0$;
  \item\label{3} $\hat\zeta_7=\hat\xi_7+w_2=0$.
 \end{enumerate}
 Denote by $E_k$ the vector of $\RR^d$ with the $k$th entry equal to 1 and the other entries equal to zero. The following holds.
 \begin{description}[font=\rm\itshape]
  \item[Misfolding of OmpR-P.] The projection of the difference between EnvZ + OmpR$^\star$-P and EnvZ-OmpR-P onto the species of the EnvZ-OmpR signaling system is $E_2-E_6$, which is orthogonal to $\hat\zeta$ by \ref{2}. The projection of the difference between OmpR-P + C and OmpR$^\star$-P + C is $E_7$, which is orthogonal to $\hat\zeta$ by \ref{3}.
  \item[Production and degradation of chaperons.] By assumption, any chemical reaction $y\to y'$ involved in the production and degradation of chaperones does not consume or produce any chemical species of the EnvZ-OmpR signaling system. Hence, $\pi(y'-y)=0$, which is orthogonal to $\hat\zeta$.
  \item[External regulation of EnvZ.] The difference between EnvZ and 0 is $\pm E_2$, which is orthogonal to $\hat\zeta$ by \ref{1}
  \item[Transcription of OmpC and OmpF.] We assume that the transcription only involves the species OmpR-P, out of all the species in the EnvZ-OmpR osmoregulation system. Hence, for all the reactions $y\to y'$ involved in the transcription, either $\pi(y'-y)=0$ or $\pi(y'-y)$ is a multiple of $E_7$. In either case, $\pi(y'-y)$ is orthogonal to $\hat\zeta$.
 \end{description}
 It follows from Corollary~\ref{thm:extension} that the species OmpR-P is ACR in the reaction system of Figure~\ref{fig:larger}. Moreover, its ACR value is still given by \eqref{eq:calculation_ex}, as long as a positive steady state exists. Note that Corollary~\ref{thm:extension} could be applied even if not all chemical reactions are known, and even if the model is not mass-action. It is also worth noting that the deficiency of the model is not known, due to the lack of information on the precise reactions constituting the network, but is certainly greater than 1. Indeed, the deficiency of the sub-system constituted by the EnvZ-OmpR osmoregulation system and by the misfolding of OmpR-P is 2, and the deficiency of a system is necessarily greater than or equal to the deficiency of any sub-system \cite[Lemma 5]{CW:poisson}.

  \begin{figure}[h!]
\begin{center}
 \begin{tcolorbox}[colback=olive!5!white,colframe=olive!50!white]
  EnvZ-OmpR osmoregulation system, as described in Figure~\ref{fig:envz} 
  
 \DrawLine{olive!50!white}
  Misfolding of sensory response protein OmpR-\Ph, and recovery through the action of chaperones:
$$
\schemestart
EnvZ-OmpR-\Ph \arrow{->[*0{$\kappa_{12}$}]} EnvZ \+ {\color{teal} OmpR$^\star$}-\Ph
\arrow(@c1.south east--.north east){0}[-90,.3]
{\color{teal} OmpR$^\star$}-\Ph \+ C \arrow{->[*0{$\kappa_{13}$}]} OmpR-\Ph \+ C
\schemestop
$$
\DrawLine{olive!50!white}
  Production and degradation of chaperones:
\begin{center}
\begin{tikzpicture}
 \node[draw, fill=white, shape = ellipse] (chap_reg) {Gene regulatory network};
 \node[right of = chap_reg, node distance=11em] (C) {C};
 \node[right of = C, node distance=4em] (0) {0};
 \draw [->, dashed] (chap_reg) -- (C);
 \draw [->] (C) -- node[above]{$u_C(t)$} (0);
\end{tikzpicture}
\end{center}
\DrawLine{olive!50!white}
  External regulation of the sensor-transmitter protein EnvZ:
$$
\schemestart
0 \arrow{<=>[*0{$u_1(t)$}][*0{$u_2(t)$}]} EnvZ 
 \schemestop
$$
\DrawLine{olive!50!white}
  Transcription of OmpC and OmpF:
\begin{center}
\begin{tikzpicture}
 \node[draw, fill=white, shape = ellipse] (OmpR) {OmpR-\Ph\ binds to DNA promoter};
 \node[below of = OmpR, node distance = 2em] (foo) {};
 \node[right of = foo, node distance = 8em] (OmpC) {OmpC};
 \node[left of = foo, node distance = 8em] (OmpF) {OmpF};
 \draw [->, dashed] (OmpR) -- (OmpC);
 \draw [->, dashed] (OmpR) -- (OmpF);
\end{tikzpicture}
\end{center}
 \end{tcolorbox}
 \caption{Reaction system including the EnvZ-OmpR osmoregulation system depicted in Figure~\ref{fig:envz} as a sub-system. Parts of the system depicted here are unknown, specifically no model is given for the production of chaperons or for the transcription of the outer membrane porins OmpF and OmpC.}
 \label{fig:larger}
 \end{center}
\end{figure}
 
 \section{Discussion}
 
 We have shown in Theorem~\ref{thm:lcif} that a linear CI always exists for a family of ACR systems, and that this family strictly includes the models studied in \cite{SF:ACR}. The result, a more general version of which is proven in the Supplementary Material, has three main consequences: first, it provides an easy way to calculate the ACR value of ACR species. Secondly, as expressed in Theorem~\ref{thm:control}, the presence of a linear CI implies that the system is robust to arbitrary disturbances that do not vanish over time, under certain conditions. This fact can be naturally exploited to design perfect insulators, which are able to reject loading effects originated from the downstream modules. Finally, as expressed in Corollary~\ref{thm:extension}, we are able to prove that, under certain conditions, the absolute concentration robustness of a portion of a system can be lifted to the whole model, and the ACR value of the ACR species remains unchanged.
 
 The theory we developed opens the path for future research directions. First, efficient algorithms can be designed in order to check for the existence of portions of the systems that confer absolute concentration robustness to the whole system. To this aim, the connections of $\hat\Gamma_{ij}(\kappa)$ with structural properties of the network which we show in the Supplementary Material can be useful, and theoretical results can be expanded in this direction. Second, further detailed analysis on when stability can be ensured would be welcome. Currently, in the statement of Theorem~\ref{thm:control} we cannot exclude the possibility that some species is completely consumed or indefinitely produced upon tampering with the model. Finding structural conditions able to eliminate this possibility would be a nice and useful contribution.
 
 As a final remark, we think the study of stochastically modeled systems that satisfy the assumptions of Theorem~\ref{thm:lcif} would be interesting and fruitful. Stochastic models of reaction systems are tipycally used when few molecules of certain chemical species are available \cite{AK:book, ET:book}. It is proven in \cite{AEJ:ACR} that systems satisfying the assumptions of Theorem~\ref{thm:feinberg_weak}, when stochastically modeled, undergo an extinction event almost surely. As a consequence, the desirable robustness properties of the ACR systems studied in \cite{SF:ACR} are completely destroyed in a low molecule copy-number regime. As an example, the model depicted in (\ref{eq:toymodel}) undergoes an almost sure extinction of the chemical species $B$ when stochastically modeled, regardless the initial conditions. This is caused by the fact that all the molecules of $B$ can be consumed by the reaction $B\to A$, before the occurrence of a reaction $A+B\to 2B$. Robustness at finite time intervals of some stochastically modeled ACR systems is recovered, but only in a multiscale limit sense \cite{ACK:ACR}. Moreover, it is shown in \cite{AC:ACR} that absolute concentration robustness of the deterministic model does not necessarily imply an extinction event in the corresponding stochastic model, but the connection is still largely unexplored. The results developed in the present paper can help in this direction: consider again (\ref{eq:toymodel}). The extinction of species $B$ cannot occur if production of $B$ is included in the model as in (\ref{eq:toymodel_modified}), or as in
 \begin{equation}\label{eq:toymodel_modified_2}
 \begin{split}
  \schemestart
   $A+B$\arrow{->[$\kappa_1$]}$2B$
   \arrow(@c1.south east--.north east){0}[-90,.25]
   $B$\arrow{->[$\kappa_2$]}$A$
   \arrow(@c3.south east--.north east){0}[-90,.25]
   $0$\arrow{<=>[$\kappa_3$][$\kappa_4$]}$B$
  \schemestop
  \end{split}
 \end{equation}
 At the same time, it follows from Theorem~\ref{thm:control} that the stability properties of the species $A$ are maintained both in (\ref{eq:toymodel_modified}) and in (\ref{eq:toymodel_modified_2}), when deterministically modeled. In particular, the concentration of the species $A$ still converges to the value $\kappa_2/\kappa_1$. It would be interesting to study if in this and in similar cases some form of absolute concentration robustness arise in the long-term dynamics of the stochastic models as well.
 
 \section*{Acknowledgement}
 This project has received funding from the European Research Council (ERC) under the European Union’s Horizon 2020 research and innovation programme grant agreement no. 743269 (CyberGenetics project).
 
 \bibliographystyle{plain}
 \bibliography{bib}

\newpage
\begin{center}
    {\bfseries\Large Supplementary Material}
\end{center}
\appendix
\renewcommand{\thesection}{\Alph{section}}
\setcounter{equation}{0}
\numberwithin{equation}{section}
\renewcommand{\theequation}{\thesection.\arabic{equation}}

\section{Notation}\label{sec:notation}

\subsection{General notation}

We will denote by $\RR$, $\RR_{>0}$, and $\RR_{\geq0}$ the real, positive real, and non-negative real numbers, respectively. Similarly, we will denote by $\ZZ$, $\ZZ_{>0}$, and $\ZZ_{\geq0}$ the integer, positive integer, and non-negative integer numbers, respectively. Given a real number $a$, we will denote its absolute value by $|a|$.

We denote by $e_i$ the vector that has 1 in its $i$th entry and 0 in all other entries. The dimension of such vector will be clear from the context, and if not it will be made explicit. Given two real vectors $v$, $w$ of the same length $n$, we denote their scalar product by $\scal{v}{w}$, and we use the shorthand notation
$$v^w=\prod_{i=1}^n v_i^{w_i},$$
where $0^0$ is considered to be 1. We will denote the euclidean norm of $v$ by $\|v\|$.

Given two subsets $V, W\subseteq\RR^n$, we denote their direct sum by
$$V\oplus W=\{v+w\,:\,v\in V, w\in W\}.$$
Moreover, given a vector $v\in\RR^n$, we define
$$v+W=\{v+w\,:\,w\in W\}.$$

\subsection{Reaction network terminology}

\subsubsection{Standard notation}
A reaction network $\G$ is a triple $\{\Sp, \Cx,\Rc\}$, where
\begin{itemize}
 \item $\Sp$ is a finite ordered set of $d$ different symbols, called \emph{species};
 \item $\Cx$ is a finite ordered set of $m$ linear combinations of species with non-negative integer coefficients, referred to as \emph{complexes} and identified with vectors in $\ZZ_{\geq0}^d$;
 \item $\Rc$ is a finite ordered set of elements of $\Cx\times\Cx$, referred to as \emph{reactions}.
\end{itemize}
 We will denote by $X_i$ the $i$th species of $\Sp$, and by $y_j$ the $j$th complex of $\Cx$, for $1\leq i\leq d$ and $1\leq j\leq m$. A reaction $(y_i,y_j)$ will be denoted by $y_i\to y_j$, and for any $y_i\in\Cx$ we assume there is no reaction of the form $y_i\to y_i$.

As mentioned, a complex $y_i$ can be regarded as a vector of $\ZZ^d_{\geq0}$. Specifically, this is done by considering the $j$th entry $y_{ij}$ as the coefficient of $y_i$ relative to the $j$th species. In this regards, for any vector $v\in\RR^d$ we will denote by $\supp v$ the subset of species such that
$$X_i\in\supp v\text{ if and only if }v_i\neq0.$$
Similarly, for any vector $w\in\RR^m$ we will denote by $\supp w$ the subset of complexes such that
$$y_i\in\supp v\text{ if and only if }v_i\neq0.$$

We denote by
$$\St=\spann_{\RR}\{y_j-y_i\,:\,y_i\to y_j\in\Rc\},\qquad \Cons=\{h\in\RR^d\,:\,\scal{h}{y_j-y_i}=0\text{ for all }y_i\to y_j\in\Rc\}.$$
The subspace $\St$ is called \emph{stoichiometric subspace}, and the elements of $\Cons$ are called \emph{conservation laws}.

A choice of kinetics for a reaction network is a set of (time-dependent) \emph{rate functions} $\lambda_{ij}\colon\RR^{d}_{\geq0}\times\RR_{\geq0}\to \RR_{\geq0}$ for all $1\leq i,j\leq m$, such that $\lambda_{ij}$ is a zero function if and only if $y_i\to y_j\notin\Rc$. A reaction network with a choice of kinetics $\Sy=(\Sp, \Cx, \Rc,\{\lambda_{ij}\}_{1\leq i,j\leq m})$ is termed \emph{reaction system}. A reaction system is associated with the system of differential equations
$$\frac{d}{dt}x(t)=\sum_{1\leq i,j\leq m}(y_j-y_i) \lambda_{ij}(x(t),t)\quad \text{for all }t\in\RR_{\geq0}.$$
We note that rate functions are commonly intended to not depend on time, but we consider this more general setting in this paper.

A reaction system is called \emph{mass-action system}, and denoted by $\Sy=(\Sp, \Cx, \Rc,\kappa)$, if there is a matrix $\kappa\in\RR^{m\times m}_{\geq0}$ such that
$$\lambda_{ij}(x,t)=\kappa_{ij}x^{y_i}\quad\text{for all }1\leq i,j\leq m, x\in\RR^{d}_{\geq0}.$$
In this case, the constants $\kappa_{ij}$ are termed \emph{rate constants}. 
Define $\Lambda(x)\in\RR^{d}_{\geq0}$ by
$$\Lambda_i(x)=x^{y_i}\quad\text{for all }1\leq i\leq m, x\in\RR^{d}_{\geq0},$$
and define the $m\times m$ matrix $A(\kappa)$ by
$$A(\kappa)_{ij}=\begin{cases}
                 \kappa_{ji}&\text{if }i\neq j\\
                 -\sum_{l=1}^m \kappa_{il} &\text{if }i=j
                \end{cases}
$$
Finally, let $Y$ be the $d\times m$ matrix $Y$ with entries $Y_{ij}=y_{ji}$. Then, for mass-action sytems we have
$$\frac{d}{dt}x(t)=\sum_{1\leq i,j\leq m}(y_j-y_i) \kappa_{ij}x(t)^{y_i}=Y A(\kappa) \Lambda(x(t))\quad\text{for all }t\in\RR_{\geq0}.$$

The directed graph $\{\Cx, \Rc\}$ is called \emph{reaction graph}. Reaction systems are often presented through the reaction graph, where indication on the reaction rates are written on top of the arrows that correspond to the related reaction. Specifically, the arrow corresponding to $y_i\to y_j$ is labeled with
$$\frac{\lambda_{ij}(x,t)}{x^{y}},$$
which corresponds to the rate constants for reaction rates of mass-action type. As an example, see \eqref{eq:toymodel} and \eqref{eq:toymodel_modified_graph} in the main text, or Figure~\ref{fig:envz}. We denote by $\ell$ the number of connected components of the reaction graph. We define the \emph{deficiency} as the number
$$\delta=m-\ell-\dim \St.$$
A \emph{terminal component} of the reaction graph is a set of complexes $\mathcal{T}\subseteq\Cx$, such that
\begin{itemize}
 \item if $y\in\mathcal{T}$ and there is a directed path in the reaction graph from $y$ to another complex $y'$, then $y'\in\mathcal{T}$;
 \item for any two complexes $y,y'\in\mathcal{T}$ with $y\neq y'$, there is a directed path  from $y$ to $y'$, and a directed path from $y'$ to $y$.
\end{itemize}
Let $\tau$ be the number of terminal components, and denote by $\mathcal{T}^1, \mathcal{T}^2, \dots, \mathcal{T}^\tau$ the different terminal components of the network. Since each connected component contains at least one terminal component, we have $\tau\geq\ell$. We say a complex is \emph{terminal} if it is contained in a terminal component, and we say that a complex is \emph{non-terminal} otherwise. As an example, the terminal components in \eqref{eq:toymodel} are $\{2B\}$ and $\{A\}$, hence the terminal complexes are $2B$ and $A$.
 
\subsection{Absolute concentration robustness}

\begin{definition}
 Consider a reaction system $\Sy=(\Sp, \Cx, \Rc, \{\lambda_{ij}\}_{1\leq i,j\leq m})$. A species $X_n\in\Sp$ is said to be \emph{absolutely concentration robust} (ACR) if there exists $q\in\RR_{>0}$ such that $c_n=q$ for all positive steady state $c$ of $\Sy$. In this case, $q$ is referred to as the \emph{ACR value} of $X_n$. Finally, if an ACR species exists, then the reaction system $\Sy$ is said to be \emph{ACR}.
\end{definition}

\begin{remark}\label{rem:degenerate_ACR}
 By definition, if a reaction system has no positive steady states, or a unique one, then all the species are ACR. We can call this a \emph{degenerate case}, since, especially if no positive steady states exist, the concept of robustness is lost.
\end{remark}

\subsection{Control Theory terminology}

Consider a differential equation of the form
$$ \frac{d}{dt} x(t)=f(x(t))+g(x(t), u(t))\quad\text{for all }t\in\RR_{\geq0},$$
where $x\colon\RR_{\geq0}\to \RR^d$ and $u\colon\RR_{\geq0}\to\RR^{n_u}$ for some $n_u\in\ZZ_{>0}$, and $f$ and $g$ are differentiable functions. The function $u$ is called the \emph{input of the system}. Further, we define \emph{output of the system} the quantity $z(t)=a(x(t))$, for some differentiable function $a$, with $a\colon\RR^{d}\to\RR^{n_z}$ and $n_z\in\ZZ_{>0}$. Usually, one needs to find an appropriate function $u$ such that $z$ is close to a desired level $\overline{z}\in\RR^{n_z}$, either on average or for $t\to\infty$. To this aim, the existence of a function $\phi\colon \RR^{n_x}\to\RR$ such that
$$\frac{d}{dt}\phi(x(t))=z(t)-\overline{z}$$
is of high importance, an is called an \emph{integral feedback} (IF) \cite{doyle:feedback, astrom:feedback}. If a function $\tilde\phi\colon \RR^{n_x}\to\RR$ satisfies
$$\frac{d}{dt}\tilde\phi(x(t))=r(x(t))\Big(z(t)-\overline{z}\Big)$$
for some differentiable function $r\colon \RR^{n_x}\to\RR$, then $\tilde\phi$ is called a \emph{constrained integral feedback} (CIF) \cite{XD:robust}.

In the setting of ACR species, $z(t)$ is usually the concentration of the ACR species at time $t$, and $\overline{z}$ is its ACR value.

\section{Known results}

The following result appears in the Appendix of \cite{FH:1977}.
\begin{theorem}\label{thm:support}
 All the vectors in $\ker A(\kappa)$ have support in the terminal complexes. Specifically, there exists a basis $\{\chi^1(\kappa), \chi^2(\kappa), \dots, \chi^\tau(\kappa)\}$ of $\ker A(\kappa)$ such that
 $\supp \chi^i(\kappa)=\mathcal{T}^i$ for all $1\leq i\leq m$.
\end{theorem}
The following result can be deduced from Section 6 in \cite{FH:1977}.
\begin{theorem}\label{thm:deficiency}
 We have
 $$\delta\geq\dim\Big(\ker Y\cap \Image A(\kappa)\Big)=\dim\ker YA(\kappa) - \dim\ker A(\kappa).$$
\end{theorem}
Finally, the following result is proven in the supplementary material of \cite{SF:ACR}, where it is stated as Theorem S3.15.
\begin{theorem}\label{thm:feinberg_strong}
 Consider a mass-action system $(\Sp, \Cx, \Rc,\kappa)$. Assume the following conditions hold:
 \begin{enumerate}
  \item $y_i$ and $y_j$ are non-terminal complexes;
  \item the deficiency is 1;
 \end{enumerate}
 Then, there exists $q\in\RR_{>0}$ such that $c^{y_j-y_i}=q$ for all positive steady states $c$. 
\end{theorem}

For convenience, we restate here the weaker version given in the main text:
\thmfeinbergweak*
Note how Theorem~\ref{thm:feinberg_weak} is a straightforward consequence of Theorem~\ref{thm:feinberg_strong}: if the two complexes only differ in the $n$th entry, then
$$c^{y_j-y_i}=c_n^{(y_j-y_i)_n}=q$$
for all positive steady states $c$, hence $X_n$ is ACR.

\begin{remark}
 In the original statement of Theorem~\ref{thm:feinberg_strong}, the additional hypothesis that a positive steady state exists is made. In fact, under the condition of Theorem~\ref{thm:feinberg_strong}, the existence of positive steady states is not guaranteed. However, if no positive steady state exists then the conclusion of the theorem holds trivially. Hence, our reformulation holds correct. See Remark~\ref{rem:degenerate_ACR} for the formal relationship between ACR species and the lack of positive steady states.
\end{remark}

\section{Calculation and structural properties of $\Gamma_{ij}(\kappa)$}

Consider a mass-action system $\Sy=(\Sp, \Cx, \Rc, \kappa)$. For any $1\leq i,j\leq m$ with $i\neq j$ define the set
\begin{equation}\label{eq:GammaAp}
 \Gamma_{ij}(\kappa)=\left\{\gamma\in\RR^{d+1}\,:\,
\begin{pmatrix}
A(\kappa)^\top Y^\top\,|\, e_i
\end{pmatrix}\gamma
=e_j\right\},
\end{equation}
and for any real vector $v$ of length larger than $d$, let $\pi_d(v)$ be its projection onto the first $d$ components. We define
\begin{equation}\label{eq:gamma_hat}
\hat  \Gamma_{ij}(\kappa)=\left\{\hat\gamma\in\RR^d\,:\,\hat\gamma=\pi_d(\gamma)\text{ for some }\gamma\in \Gamma_{ij}(\kappa)\right\}.
\end{equation}
The set $\Gamma_{ij}(\kappa)$ can be computed by first calculating a basis for 
\begin{equation}\label{eq:Psi}
 \Psi_{ij}(\kappa)=\ker \begin{pmatrix}
A(\kappa)^\top Y^\top\,|\, e_i \,|\, e_j
\end{pmatrix}.
\end{equation}
Note that this can be easily and quickly done by using a programming language which is able to deal with symbolic linear algebra, such as Matlab. The following holds
\begin{proposition}\label{prop:calculating_Gamma}
 Consider a mass-action system $\Sy=(\Sp, \Cx, \Rc,\kappa)$, and $1\leq i,j\leq m$ with $i\neq j$. Let $\{\psi^1, \psi^2, \dots, \psi^k\}$ be a basis for $\Psi_{ij}(\kappa)$, as defined in \eqref{eq:Psi}. Then, $\Gamma_{ij}(\kappa)$ (and consequently $\hat \Gamma_{ij}(\kappa)$) is non-empty if and only if $\psi^n_{d+2}\neq 0$ for some $1\leq n\leq k$.  If this is the case, then
 \begin{align*}
  \Gamma_{ij}(\kappa)&=\left\{-\frac{1}{\sum_{n=1}^k a_n\psi^n_{d+2}}\sum_{n=1}^k a_n\pi_{d+1}(\psi^n)\,:\,a_1,a_2,\dots,a_k\in\RR\text{ and }\sum_{n=1}^k a_n\psi^n_{d+2}\neq 0\right\},\\
  \intertext{where $\pi_{d+1}\colon\RR^{d+2}\to \RR^{d+1}$ is the projection onto the first $d+1$ components, and}
  \hat\Gamma_{ij}(\kappa)&=\left\{-\frac{1}{\sum_{n=1}^k a_n\psi^n_{d+2}}\sum_{n=1}^k a_n\pi_{d}(\psi^n)\,:\,a_1,a_2,\dots,a_k\in\RR\text{ and }\sum_{n=1}^k a_n\psi^n_{d+2}\neq 0\right\}.
 \end{align*}
\end{proposition}
\begin{proof}
 First, note that $\gamma\in\Gamma_{ij}(\kappa)$ if and only $(\gamma, -1)\in\Psi_{ij}(\kappa)$. This can be easily deduced by the definitions \eqref{eq:GammaAp} and \eqref{eq:Psi}. The proof simply follows from this equivalence, and from the definition of $\hat \Gamma_{ij}(\kappa)$ given in \eqref{eq:gamma_hat}. 
\end{proof}

In Section~\ref{sec:models} we will use Proposition~\ref{prop:calculating_Gamma} to calculate $\hat\Gamma_{ij}(\kappa)$ for the main examples discussed in the main text. We will see that some vectors in the basis of $\Psi_{ij}(\kappa)$ do not depend on the particular choice of rate constants. As a consequence, some dynamical properties implied by the theory developed in this paper will only depend on the structure of the model rather than on a fine tuning of the parameters, which is desirable. In the following result, we explicitly derive structural properties of $\hat\Gamma_{ij}(\kappa)$.

\begin{proposition}\label{prop:structure_Gamma}
 Consider a mass-action system $\Sy=(\Sp, \Cx, \Rc,\kappa)$, and $1\leq i,j\leq m$ with $i\neq j$. Assume that $\hat\Gamma_{ij}(\kappa)$ is non-empty, and let $\hat\gamma\in\hat\Gamma_{ij}(\kappa)$. Then, 
 \begin{equation}\label{eq:containment}
 \hat\gamma+\Cons\subseteq\hat\gamma+\ker A(\kappa)^\top Y^\top\subseteq \hat\Gamma_{ij}(\kappa). 
 \end{equation}
 Moreover, if the mass-action system has a steady state $c$ with $c^{y_i}>0$, then
 \begin{equation}\label{eq:same_q}
  \gamma'_{d+1}=c^{y_j-y_i}\quad\text{for all }\gamma'\in\Gamma_{ij}(\kappa)
 \end{equation}
and
 \begin{equation}\label{eq:equality1}
 \hat\Gamma_{ij}(\kappa)=\hat\gamma+\ker A(\kappa)^\top Y^\top. 
 \end{equation}
Finally, if each connected component of the reaction graph contains exactly one terminal component, then 
 \begin{equation}\label{eq:equality2}
 \\ker A(\kappa)^\top Y^\top=\Cons. 
 \end{equation}
\end{proposition}

\begin{proof}
 First, we have that for any $\hat \gamma\in \hat\Gamma_{ij}(\kappa)$
 \begin{equation}\label{eq:containment_ker}
  \hat\gamma+\ker A(\kappa)^\top Y^\top\subseteq\hat\Gamma_{ij}(\kappa).
 \end{equation}
 Indeed, for any $\hat\gamma\in\hat\Gamma_{ij}(\kappa)$, there exists $\gamma\in\Gamma_{ij}(\kappa)$ with $\hat\gamma=\pi_d(\gamma)$. Then, for any for any $v\in\ker A(\kappa)^\top Y^\top$ we have
 $$\begin{pmatrix}
A(\kappa)^\top Y^\top\,|\, e_i
\end{pmatrix}\left(\gamma+\begin{pmatrix} v \\ 0\end{pmatrix}\right)
=\begin{pmatrix}
A(\kappa)^\top Y^\top\,|\, e_i
\end{pmatrix}\gamma + \begin{pmatrix}A(\kappa)^\top Y^\top v \\ 0\end{pmatrix}
=e_j.$$
Hence,
$$\gamma+\begin{pmatrix} v \\ 0\end{pmatrix}\in\Gamma_{ij}(\kappa)$$
which implies that $\hat\gamma+v\in\hat\Gamma_{ij}(\kappa)$ and proves \eqref{eq:containment_ker}. Since  for any $h\in\Cons$ and any $1\leq n\leq m$
 $$\Big(A(\kappa)^\top Y^\top h\Big)_n=\sum_{l=1}^m \scal{h}{y_l-y_n}\kappa_{nl}=0,$$
 it follows
 \begin{equation}\label{eq:cons_perp}
  \Cons\subseteq \ker A(\kappa)^\top Y^\top.
 \end{equation}
 \eqref{eq:containment} follows from \eqref{eq:containment_ker} and \eqref{eq:cons_perp}.
 
 Now assume that the mass-action system has a steady state $c$, with $c^{y_i}>0$. Then, 
 $$0=\frac{d}{dt}\scal{\hat\gamma}{x(t)}|_{x(t)=c}=\hat\gamma^\top YA(\kappa)\Lambda(c)=(e_j-\gamma_{d+1}e_i)\Lambda(c)=c^{y_j}-\gamma_{d+1}c^{y_i}.$$
 Since $c^{y_i}>0$, then necessarily $\gamma_{d+1}=c^{y_j-y_i}$. For the same argument, for any other $\gamma'\in\Gamma_{ij}(\kappa)$, $\gamma'_{d+1}=\gamma_{d+1}=c^{y_j-y_i}$. Hence, \eqref{eq:same_q} holds. It follows that any $\gamma'\in\Gamma_{ij}(\kappa)$ is of the form
 $$\gamma'=\gamma+\begin{pmatrix} v \\ 0\end{pmatrix}$$
 for some $v\in\RR^{d}$. Moreover, since $\gamma'\in\Gamma_{ij}(\kappa)$
 $$e_j=\begin{pmatrix}
A(\kappa)^\top Y^\top\,|\, e_i
\end{pmatrix}\gamma'
=\begin{pmatrix}
A(\kappa)^\top Y^\top\,|\, e_i
\end{pmatrix}\gamma + \begin{pmatrix}A(\kappa)^\top Y^\top v \\ 0\end{pmatrix}= e_j+\begin{pmatrix}A(\kappa)^\top Y^\top v \\ 0\end{pmatrix}.$$
Hence, necessarily $v\in\ker A(\kappa)^\top Y^\top$, which implies
$$\hat\Gamma_{ij}(\kappa)\subseteq\hat\gamma+\ker A(\kappa)^\top Y^\top.$$
The latter, together with \eqref{eq:containment_ker}, implies \eqref{eq:equality1}.
 
To conclude the proof, we need to show that if each connected component of the reaction graph contains exactly one terminal component (i.e.\ if $\ell=\tau$), then \eqref{eq:equality2} holds. This follows from \eqref{eq:cons_perp} and
 \begin{align*}
  \dim \Cons&= d-\dim\St= d+\ell+\delta-m\\
  &=d+\tau+\delta-m\geq d-m+\dim\ker YA(\kappa)\\
  &=\dim \ker A(\kappa)^\top Y^\top,
 \end{align*}
 where we utilized Theorems~\ref{thm:support} and \ref{thm:deficiency} for the forth equality.
\end{proof}

As a consequence, we have the following.

\begin{corollary}\label{cor:structural}
   Consider a mass-action system, and assume that $y_i$ and $y_j$ are two complexes for which $\hat\Gamma_{ij}(\kappa)$ is non-empty. Let $\{v_p\}_{p=1}^{H}$ be a basis for $\Cons$, where $H=d-\dim\St$. Moreover, let $X_{l_1}, X_{l_2}, \dots, X_{l_n}\in\Sp$ such that the rank $H\times n$ matrix $V$ has rank $n$, where $V_{pq}=v_{p l_q}$ for all $1\leq p\leq H$ and $1\leq q\leq n$.
   Then, there exists $\hat\gamma\in \hat\Gamma_{ij}(\kappa)$ such that $\scal{\hat\gamma}{e_p}=0$ for all $1\leq p\leq n$.
\end{corollary}
\begin{proof}
 Let $\hat\gamma^\star\in \hat\Gamma_{ij}(\kappa)$. It follows from Proposition~\ref{prop:structure_Gamma} that for all $w\in \RR^H$
 $$\hat\gamma^\star+\sum_{p=1}^H w_pv_p\in \hat\Gamma_{ij}(\kappa).$$
 Hence, the proof is concluded by choosing $w$ such that
 $$w^\top V=-(\hat\gamma^\star_{l_1},\hat\gamma^\star_{l_2},\dots,\hat\gamma^\star_{l_n}),$$
 which is possible because the $H\times n$ matrix $V$ has rank $n$. 
\end{proof}

Further useful structural property of $\hat\Gamma_{ij}(\kappa)$ are following. Such properties are useful while checking whether the results of this paper can be applied, and can be used by an algorithm designed to this aim. 
Before stating the structural results, it is convenient to prove the following lemma.

\begin{lemma}\label{lem:structural}
 Consider a mass-action system $\Sy=(\Sp, \Cx, \Rc,\kappa)$, and $1\leq i,j\leq m$ with $i\neq j$. Then, $\hat\Gamma_{ij}(\kappa)$ is non-empty if and only if $e_j-\gamma_{d+1}e_i\in(\ker YA(\kappa))^\top$ for some $\gamma_{d+1}\in\RR$. Moreover, if that is the case and if the mass-action system has a steady state $c$ with $c^{y_i}>0$ and $c^{y_j}>0$, then necessarily $\gamma_{d+1}>0$.
\end{lemma}
\begin{proof}
 By definition of $\hat\Gamma_{ij}(\kappa)$ given in \eqref{eq:gamma_hat}, $\hat\Gamma_{ij}(\kappa)$ is non-empty if and only if $\Gamma_{ij}(\kappa)$ is non-empty. By \eqref{eq:GammaAp}, $\Gamma_{ij}(\kappa)$ is non-empty if and only if there exists $\gamma_{d+1}\in\RR$ such that $e_j-\gamma_{d+1}e_i$ is in $\Image A(\kappa)^\top Y^\top$. By the fundamental theorem of linear algebra, the latter holds if and only if $e_j-\gamma_{d+1}e_i$ is orthogonal to $\ker YA(\kappa)$.
 
 To conclude the proof, we note that if the mass-action system has a steady state $c$ with $c^{y_i}>0$ and $\hat\Gamma_{ij}(\kappa)$ is non-empty, then it follows by \eqref{eq:same_q} in Proposition~\ref{prop:structure_Gamma} that the quantity $\gamma_{d+1}$ discussed in the first part of the proof is necessarily equal to $c^{y_j-y_i}>0$.
\end{proof}

\begin{proposition}\label{prop:empty_gamma_terminal}
 Consider a mass-action system $\Sy=(\Sp, \Cx, \Rc,\kappa)$, and $1\leq i,j\leq m$. Assume that one of the following holds:
 \begin{itemize}
  \item $y_i$ is non-terminal and $y_j$ is terminal;
  \item $y_i$ and $y_j$ are in two different terminal components.
 \end{itemize}
 Then, $\hat\Gamma_{ij}(\kappa)$ is empty.
\end{proposition}
\begin{proof} 
 By assumption, $y_j$ is terminal. By Theorem~\ref{thm:support}, there exists a vector
 $$\chi\in\ker A(\kappa)\subseteq \ker YA(\kappa)$$
 such that $\supp \chi$ is the terminal component containing $y_j$. Hence, if $y_i$ is non-terminal or if $y_i$ is in a different terminal component than $y_j$, we have that for all $\gamma_{d+1}\in\RR$
 $$\scal{\chi}{e_j-\gamma_{d+1}e_i}=\chi_j\neq0.$$
 It follows from Lemma~\ref{lem:structural} that $\hat\Gamma_{ij}(\kappa)$ is empty, which concludes the proof. 
\end{proof}

\begin{proposition}\label{prop:empty_gamma_symmetric}
 Consider a mass-action system $\Sy=(\Sp, \Cx, \Rc,\kappa)$, and $1\leq i,j\leq m$ with $i\neq j$. Assume that a steady state $c$ with $c^{y_i}>0$ and $c^{y_j}>0$ exists. Then, $\hat\Gamma_{ij}(\kappa)$ is empty if and only if $\hat\Gamma_{ji}(\kappa)$ is empty.
\end{proposition}
\begin{proof}
 By the symmetric role of $i$ and $j$, it suffices to prove that if $\hat\Gamma_{ij}(\kappa)$ is non-empty, then necessarily $\hat\Gamma_{ij}(\kappa)$ is non-empty.
 
 By Lemma~\ref{lem:structural}, if $\hat\Gamma_{ij}(\kappa)$ is non-empty then there exists $\gamma_{d+1}\in\RR$ with $e_j-\gamma_{d+1}e_i\in(\ker YA(\kappa))^\top$, and $\gamma_{d+1}>0$. Hence, 
 $$e_i-\frac{1}{\gamma_{d+1}}e_j\in(\ker YA(\kappa))^\top,$$
 hence by Lemma~\ref{lem:structural} $\hat\Gamma_{ji}(\kappa)$ is non-empty, which is what we wanted to show.
\end{proof}

\section{Proofs of the results stated in the main text}

In this section, we state and prove more general versions of the theorems presented in the main text. 

\subsection{Existence and characterization of a CIF}

 The following result is stated in the main text.
 \thmfeinbergimpliescond*
 
The following more general result holds, from which Theorem~\ref{thm:feinberg_implies_cond} can be immediately deduced.
\begin{theorem}
 Consider a mass-action system, and assume the following holds:
 \begin{enumerate}
  \item\label{part_non_terminal_complexes} $y_i$ and $y_j$ are two distinct non-terminal complex;
  \item\label{part_deficiency_1} the deficiency is 1;
  \item\label{part_existence_of_pos_eq} a steady state $c$ with $c^{y_i}\neq0$ exists.
 \end{enumerate}
 Then, $\hat\Gamma_{ij}(\kappa)$ is non-empty.
\end{theorem}
 \begin{proof}
  By condition~\ref{part_existence_of_pos_eq}, $c$ is a steady state, hence the vector $\Lambda(c)$ is in $\ker Y A(\kappa)$. Moreover, $\Lambda_i(c)=c^{y_i}\neq0$. By condition~\ref{part_non_terminal_complexes}, $y_i$ is a non-terminal complex. Hence, by combining $\Lambda_i(c)\neq0$ with Theorem~\ref{thm:support}, it follows that $\Lambda(c)$ is not in $\ker A(\kappa)$. Since by assumption $\delta=1$, it follows from Theorem~\ref{thm:deficiency} that
 \begin{equation}\label{eq:ker_direct_sum}
  \ker Y A(\kappa)=\spann_\RR\{\Lambda(c)\}\oplus\ker A(\kappa).
 \end{equation}
 Define
 $$v=e_j-c^{y_j-y_i}e_i=e_j-\frac{\Lambda_j(c)}{\Lambda_i(c)}e_i.$$
 Clearly, $v$ is orthogonal to $\Lambda(c)$. Moreover, by condition~\ref{part_non_terminal_complexes}, both $y_i$ and $y_j$ are non-terminal complexes. Hence, the vector $v$ is orthogonal to $\ker A(\kappa)$ by Theorem~\ref{thm:support}. It follows from \eqref{eq:ker_direct_sum} that $v$ is orthogonal to $\ker Y A(\kappa)$, which implies that 
 $\hat\Gamma_{ij}(\kappa)$ is non-empty by Lemma~\ref{lem:structural}.  
 \end{proof}

We now proceed to prove the following result, stated in the main text.
\thmlcif*
In order to prove the result, we will show that a more general version holds, which we state here.
\begin{theorem}\label{thm:lcif_app}
  Consider a mass-action system, and assume that $y_i$ and $y_j$ are two complexes for which $\hat\Gamma_{ij}(\kappa)$ is non-empty. Let $\gamma\in\hat\Gamma_{ij}(\kappa)$.
  Then,
  \begin{equation}\label{eq:fixed_value}
   c^{y_j-y_i}=\gamma_{d+1}
  \end{equation}
  for all steady states $c$ satisfying $c^{y_i}\neq 0$. Moreover, 
 $$\phi(x)=\sum_{i=1}^d \beta_i x_i$$
 is a linear CIF with
 \begin{equation}\label{eq:to_be_proved}
 \frac{d}{dt}\phi(x(t))=\Lambda_i(x(t)) \Big(x(t)^{y_j-y_i}-\gamma_{d+1}\Big) 
 \end{equation}
 for any initial condition $x(0)$ if and only if $\beta\in\hat\Gamma_{ij}(\kappa)$.
\end{theorem}
\begin{proof}
\eqref{eq:fixed_value} follows from \eqref{eq:same_q} in Proposition~\ref{prop:structure_Gamma}.
 
 Let $\beta\in\hat\Gamma_{ij}(\kappa)$. Then, by Proposition~\ref{prop:structure_Gamma} we have that $\beta-\pi_d(\gamma)\in\ker A(\kappa)^\top Y^\top$. Hence, for any initial condition $x(0)\in\RR^d_{\geq0}$ and any $t\in\RR_{\geq0}$,
 $$\frac{d}{dt}\scal{\beta}{x(t)}=\beta^\top Y A(\kappa)\Lambda(x(t))=\pi_d(\gamma)^\top Y A(\kappa)\Lambda(x(t))=(e_j-\gamma_{d+1} e_i)^\top\Lambda(x(t))=x(t)^{y_i}\Big(x(t)^{y_j-y_i}-\gamma_{d+1}\Big),$$
 which is \eqref{eq:to_be_proved}.
 
 Conversely, assume that \eqref{eq:to_be_proved} holds. Then, for all $x\in\RR^d_{\geq0}$ we have
 $$(\beta^\top -\pi_d(\gamma))^\top Y A(\kappa)\Lambda(x)=0.$$
 Since the entries of $\Lambda$ are linearly independent monomials on $\RR^d$, it must be $(\beta^\top -\pi_d(\gamma)) \in\ker A(\kappa)^\top Y^\top$, which implies that  $\beta\in\hat\Gamma_{ij}(\kappa)$ by Proposition~\ref{prop:structure_Gamma}.
\end{proof}

\subsection{Rejection of persistent disturbances}

We recall here the formal definition of ``oscillation'', as intended in this paper.
\begin{definition}
 We say that a function $g\colon\RR_{\geq0}\to\RR$ \emph{oscillates} around a value $\overline q\in\RR$ if for each $t\in\RR_{\geq0}$ there exist $t_+>t$ and $t_->t$ such that
 $$g(t_+)>q\quad\text{and}\quad g(t_-)<q.$$
\end{definition}

The following result holds.
\begin{theorem}\label{thm:control_app}
   Consider a mass-action system, with associated differential equation
   $$\frac{d}{dt} x(t)=f(x(t)).$$
   Assume that $y_i$ and $y_j$ are two complexes for which $\hat\Gamma_{ij}(\kappa)$ is non-empty. Then, there exists $\overline q\in\RR_{\geq0}$ such that $c^{y_j-y_i}=\overline q$ for all steady states $c$ with $c^{y_i}\neq0$.
  Consider an arbitrary function $u\colon \RR_{\geq0}^d\times\RR_{\geq0}\to \RR^d$ such that a solution to
   $$\frac{d}{dt}\tilde x(t)=f(\tilde x(t))+u(\tilde x(t),t)$$
   exists with $\tilde x(t)\geq0$ for all $t\geq0$. Assume that there exists a $\hat\gamma\in\hat\Gamma_{ij}(\kappa)$ which is orthogonal to the vector $u(x,t)$ for any $x\in\RR^d_{\geq0}$, $t\in\RR_{\geq0}$. Then, for any initial condition $\tilde x(0)\in\RR^d_{\geq0}$, at least one of the following statements holds:
  \begin{enumerate}
  \item\label{part:infty} There is a species $X_k\in\supp \hat\gamma$ such that $\limsup_{t\to\infty}\tilde x_k(t)=\infty$;
  \item\label{part:zero} There is a species $X_k\in\supp y_i$ such that $\liminf_{t\to\infty} \tilde x_k(t)=0$;
  \item\label{part:oscillation} $\tilde x(t)^{y_j-y_i}$ oscillates around $\overline q$;
  \item\label{part:convergence} $\displaystyle\lim_{t\to\infty}\int_{t}^\infty \Big|\tilde x(s)^{y_j-y_i}-\overline q\Big|ds=0$.
 \end{enumerate}
\end{theorem}
\begin{proof}
 It follows from Theorem~\ref{thm:lcif_app} that there exists $\overline q\in\RR_{\geq0}$ such that $c^{y_j-y_i}=\overline q$ for all steady states $c$ with $c^{y_i}\neq0$. Note that in this setting $c^{y_i}\neq0$ is equivalent to $c^{y_i}>0$, since the state space is limited to vectors of non-negative concentrations.
 
 Fix an initial condition $\tilde x(0)\in\RR^d_{\geq0}$. Assume that \ref{part:infty} does not occur. Hence, there exists $M\in\RR_{>0}$ such that
 $$|\scal{\hat\gamma}{\tilde x(t)}|<M\quad\text{for all }t\in\RR_{\geq0}.$$
 By the fact that $\hat\gamma$ is orthogonal to $u(x,t)$ for all $x\in\RR^d_{\geq0}$ and all $t\geq0$, and by Theorem~\ref{thm:lcif_app}, we have
 $$\frac{d}{dt}\scal{\hat\gamma}{\tilde x(t)}=\scal{\hat\gamma}{f(\tilde x(t))}=\tilde x(t)^{y_i}\Big(\tilde x(t)^{y_j-y_i}-\overline q\Big).$$
 It follows that
 \begin{equation}\label{eq:integral_bounded}
  \left|\int_0^t \tilde x(s)^{y_i}\Big(\tilde x(s)^{y_j-y_i}-\overline q\Big)ds\right|<M-|\scal{\hat\gamma}{\tilde x(0)}|\quad\text{for all }t\in\RR_{\geq0}.
 \end{equation}
 If \ref{part:zero} does not hold, then there exists $m\in\RR_{>0}$ such that $x(t)^{y_i}>m$ for all $t\in\RR_{>0}$. Together with \eqref{eq:integral_bounded}, this would imply that there exists $M^\star\in\RR_{>0}$ such that
 \begin{equation}\label{eq:difference_bound}
  \left|\int_0^t \Big(\tilde x(s)^{y_j-y_i}-\overline q\Big)ds\right|<M^\star\quad\text{for all }t\in\RR_{\geq0}.
 \end{equation}
 If \ref{part:oscillation} does not hold, then there exists $t^\star\in\RR_{>0}$ such that $\tilde x(t)^{y_j-y_i}-\overline q$ maintains the same sign for all $t>t^\star$, which together with \eqref{eq:difference_bound} implies \ref{part:convergence}. The proof is then concluded.
\end{proof}

Now we prove Theorem~\ref{thm:control}, which is stated in the main text and which we state here again for convenience. Theorem~\ref{thm:control} follows almost entirely from Theorem~\ref{thm:control_app}.
\thmcontrol*
\begin{proof}
 Assume \ref{part:infinity_or_zero} does not hold. Then, there exists $\varepsilon\in\RR_{>0}$ such that
 \begin{equation}\label{eq:boxed}
 \varepsilon\leq\|\tilde x(t)\|\leq \frac{1}{\varepsilon}\quad\text{for all }t\in\RR_{\geq0}.
 \end{equation}
 Moreover, it follows from Theorem~\ref{thm:control_app} that at least one of the following holds:
 \begin{enumerate}
  \item\label{a} $\tilde x_n(t)^{(y_j-y_i)_n}$ oscillates around $q^{(y_j-y_i)_n}$, implying that $\tilde x_n(t)$ oscillates around $q$;
  \item\label{b} $\displaystyle\lim_{t\to\infty}\int_{t}^\infty \Big|\tilde x_n(s)^{(y_j-y_i)_n}- q^{(y_j-y_i)_n}\Big|ds=0$.
 \end{enumerate}
 Since \ref{b} clearly implies \ref{part:ACR_convergence}, to complete the proof it suffices to show that \ref{a} implies $\hat\gamma\neq \hat\gamma_n e_n$. 
 
 Assume \ref{a} holds. If it were $\hat\gamma=\hat\gamma_n e_n$, then by Theorem~\ref{thm:lcif_app} and by the fact that $u(x,t)$ is orthogonal to $\hat\gamma$ for all $x\in\RR^d_{\geq0}$ we would have
 $$\hat\gamma_n\frac{d}{dt} \tilde x_n(t)=\frac{d}{dt}\scal{\hat\gamma}{\tilde x(t)}=\scal{\hat\gamma}{f(\tilde x(t))}=\tilde{x}(t)^{y_i}\Big(\tilde x_n(t)-q\Big)\quad\text{for all }t\in\RR_{\geq0}.$$
 By \eqref{eq:boxed}, there would be $m,M\in\RR_{>0}$ such that
 $$m\Big(\tilde x_n(t)-q\Big)\leq \hat\gamma_n\frac{d}{dt}\Big(\tilde x_n(t)-q\Big)\leq M \Big(\tilde x_n(t)-q\Big)\quad\text{for all }t\in\RR_{\geq0},$$
 which would in turn imply that $\tilde x_n(t)-q$ maintains the same sign for all $t\in\RR_{\geq0}$. Hence, \ref{a} could not hold and the proof is concluded. 
\end{proof}

\subsection{Inclusion in larger systems}

Loosely speaking, a subsystem $\Sy'$ of a reaction system $\Sy$ is simply the reaction system generated by a subset of the reactions of $\Sy$, with the same choice of rate functions. In order to give a formal definition, some more care is needed as the number of species involved in the subsystem may be smaller (hence the dimensions of $\Sy$ and $\Sy'$ may be different). 
\begin{definition}
 Let $\Sy=(\Sp, \Cx, \Rc, \{\lambda_{ij}\}_{1\leq i,j\leq m})$ be a reaction system with $d$ species and $m$ complexes. A \emph{subsystem} of $\Sy$ is a reaction system $\Sy'=(\Sp', \Cx', \Rc', \{\lambda'_{ij}\}_{1\leq i,j\leq m'})$ with $d'$ species and $m'$ complexes such that, up to reordering of $\Sp$ and $\Cx$,
 \begin{itemize}
  \item $\Sp'=\{X_i\in\Sp\,:\,1\leq i\leq d'\}$;
  \item $y_{in}=0$ for all $y_i\in\Cx$ with $1\leq i\leq m'$ and for all $d'<n\leq d$;
  \item $\Cx'=\{y_i'\,:\,y_i'=\pi(y_i), y_i\in\Cx, 1\leq i\leq m'\}$, where $\pi:\RR^d\to \RR^{d'}$ is the projection onto the first $d'$ components;
  \item $\Rc'\subseteq\{y'_i\to y'_j\,:\, y_i\to y_j\in\Rc\}$;
  \item for all $1\leq i,j\leq m'$ and all $x\in\RR^d_{\geq0}$
   $$\lambda'_{ij}(\pi(x))=\begin{cases}
                  \lambda_{ij}(x)&\text{if }y'_i\to y'_j\in\Rc'\\
                  0&\text{otherwise.}
                 \end{cases}$$
 \end{itemize}
\end{definition}
With a slight abuse of notation due to the potential different lenght of the complexes of $\Sy$ and $\Sy'$, we further say that $y_k\to y_l$ is a reaction of $\Sy$ but is not a reaction of $\Sy'$ if $y_k\to y_j\in\Rc$ and either $\max{k,l}>m'$ or $y'_k\to y'_l\notin\Rc'$. In this case, we write $y_k\to y_l\in\Rc\setminus\Rc'$.

The following result is stated in the main text.
\thmextension*

Here, we prove the following, stronger version.
\begin{corollary}\label{thm:extension_app}
  Consider a reaction system $\Sy=(\Sp, \Cx, \Rc, \{\lambda_{ij}\}_{1\leq i,j\leq m})$, and let $\Sy'=(\Sp', \Cx', \Rc', \{\lambda'_{ij}\}_{1\leq i,j\leq m'})$ be a sub-system. Assume that $\Sy'$ is a mass-action system with two complexes $\pi(y_i)$ and $\pi(y_j)$ for which $\hat\Gamma_{ij}(\kappa)$ is non-empty. Moreover, assume there exists $\hat\gamma\in\hat\Gamma_{ij}(\kappa)$ such that $\pi(y_l-y_k)$ is orthogonal to $\hat\gamma$ for all $y_k\to y_l\in\Rc\setminus\Rc'$.
  Then, there exists a value $\overline q\in\RR_{\geq0}$ such that
  $$(c')^{\pi(y_j-y_i)}=\overline q$$
  for all steady states $c'$ of $\Sy'$ such that $(c')^{\pi(y_i)}>0$, and
  $$c^{y_j-y_i}=\overline q$$
  for all steady states $c$ of $\Sy$ such that $c^{y_i}>0$.
 \end{corollary}
 \begin{proof}
  The existence of $\overline q\in\RR_{\geq0}$ such that
  $$(c')^{\pi(y_j-y_i)}=\overline q$$
  for all steady states $c'$ of $\Sy'$ with $(c')^{\pi(y_i)}>0$ follows from Theorem~\ref{thm:lcif_app}. 
  
  In order to prove the second part of the statement, we can write
  $$\frac{d}{dt} \pi(x(t))=f(\pi(x(t)))+u(x(t),t),$$
  where
  \begin{align*}
   f(\pi(x(t)))&=\sum_{y'_k\to y'_l\in\Rc'}(y'_l-y'_k)\kappa_{ij}\pi(x(t))^{y_i}\\
   u(x(t),t)&=\sum_{y_k\to y_l\in\Rc\setminus\Rc'}\pi(y_l-y_k)\lambda_{ij}(x(t),t).
  \end{align*}
  Note that $\hat\gamma$ is orthogonal to $u(x,t)$, for all $x\in\RR^d_{\geq0}$ and $t\in\RR_{\geq0}$. Hence, if $c$ is a steady state of $\Sy$, it follows from Theorem~\ref{thm:lcif_app} that
  \begin{align*}
   0&=\frac{d}{dt}\scal{\hat\gamma}{\pi(c)}=\scal{\hat\gamma}{f(\pi(c))}\\
   &=\pi(c)^{\pi(y_i)}\Big(\pi(c)^{\pi(y_j-y_i)}-\overline q\Big)\\
   &=c^{y_i}\Big(c^{y_j-y_i}-\overline q\Big),
  \end{align*}
  which concludes the proof.
 \end{proof}

\section{Example of an ACR system with no linear CIF}
Assume that $X_i$ is an ACR species, with ACR value $q$. It is not always possible to find a \emph{linear} combination of species whose derivative at time $t$ is of the form $r(t)(x_i(t)^\alpha-q)$, for some polynomial $r(t)$ and some real number $\alpha$. As an example, consider the following mass-action system:
 \begin{equation*}
 \begin{split}
  \schemestart
   A+B\arrow{->[$\kappa_1$]}B+C\arrow{<=>[$\kappa_2$][$\kappa_3$]}2B
   \arrow(@c1.south east--.north east){0}[-90,.25]
   B\arrow{<=>[$\kappa_4$][$\kappa_5$]}2E\arrow{->[$\kappa_6$]}2D
   \arrow(@c4.south east--.north east){0}[-90,.25]
   C\arrow{->[$\kappa_7$]}A
   \arrow(@c7.south east--.north east){0}[-90,.25]
   D\arrow{->[$\kappa_8$]}E
  \schemestop
  \end{split}
\end{equation*}
Let us order the species alphabetically, and the complexes as $(A+B, B+C, 2B, B, 2E, 2D, C, A, D, E)$. 
It is proven in \cite{AC:ACR} that the species $A$ is ACR, with ACR value
$$q=\frac{\kappa_3\kappa_7}{\kappa_1\kappa_2}.$$
In this case there is no linear combination of species whose derivative is of the form 
\begin{equation}\label{eq:desired_form}
 r(t)\left(x_1(t)^\alpha-q^\alpha\right),
\end{equation}
for some polynomial function $r(t)$ and some real number $\alpha$. Indeed, for any $\beta\in \RR^{d}$, the derivative of $\scal{\beta}{x(t)}$ is given by $\beta^\top Y A(\kappa) \Lambda(x(t))$, which in this case is a polynomial of the form 
$$w_1x_1(t)x_2(t)+w_2x_3(t)+w_3x_2(t)x_3(t)+w_4x_2(t)^2+w_5x_2(t)+w_6x_5(t)^2+w_7x_4(t),$$
for some real coefficients $w_i$. The only monomial containing $x_1(t)$ is the first one, so if we want the derivative of $\scal{\beta}{x(t)}$ to be of the form \eqref{eq:desired_form}, necessarily $w_1\neq0$, $\alpha=1$, and $r(t)=w_1x_2(t)$. We can further deduce that necessarily $w_2=w_3=w_4=w_6=w_7=0$, and $w_5=-w_1q.$
In matrix form, this is equivalent to 
$$\beta^\top Y A(\kappa)=w_1 (1,0,0,-q,0,0,0,0,0,0),$$
but this is not possible because in this case the vector on the right-hand side does not belong to the left image of $Y A(\kappa)$.

\section{Applications to the systems introduced in the main text.}\label{sec:models}

Here we use the theory developed in this work to analyze the reaction systems introduced in the main text.

\subsection{Toy example}

Consider the mass-action system
 $$
  \schemestart
   $A+B$\arrow{->[$\kappa_1$]}$2B$
   \arrow(@c1.south east--.north east){0}[-90,.25]
   $B$\arrow{->[$\kappa_2$]}$A$
  \schemestop
 $$
 Order the species alphabetically, and the complexes in the appearance order from left to right and from top to bottom, as $A+B$, $2B$, $B$, and $A$. In this case, we have
 \begin{equation}\label{eq:YAk_toyexample}
     Y=\begin{pmatrix}
      1 & 0 & 0 & 1\\
      1 & 2 & 1 & 0
     \end{pmatrix}\quad\text{and}\quad
   A(\kappa)=\begin{pmatrix}
              -\kappa_1 & 0 & 0 & 0\\
              \kappa_1 & 0 & 0 & 0\\
              0 & 0 & -\kappa_2 & 0\\
              0 & 0 & \kappa_2 & 0\\
             \end{pmatrix}
 \end{equation}
 The two non-terminal complexes differing for the entry relative to $A$ are the first and the third ones. Since the model has deficiency 1, it follows from Theorem~\ref{thm:feinberg_implies_cond} that $\hat\Gamma_{13}(\kappa)$ is non-empty. We can use Matlab to explicitly calculate it. First, in order to define $Y$ and $A(\kappa)$ in Matlab, we define the symbolic variables $\kappa_1$ and $\kappa_2$ and require they are positive real numbers via
 \begin{lstlisting}[style=Matlab-bw]
 k1=sym('k1','positive');
 k2=sym('k2','positive');
 \end{lstlisting}
 We then define the matrices \ml{Y} and \ml{Ak} as in \eqref{eq:YAk_toyexample}. We can calculate $\Psi_{13}(\kappa)$ via
 \begin{lstlisting}[style=Matlab-bw]
  e=eye(4);
  simplify(null([Ak'*Y' e(:,1) e(:,3)]))
 \end{lstlisting}
 which returns the following matrix, whose columns are a basis for $\Psi_{13}(\kappa)$:
 {\renewcommand{\arraystretch}{1.5}
 $$
 \begin{pmatrix}
  1 & -\frac{1}{\kappa_2}\\
  1 & 0\\
  0 & -\frac{\kappa_1}{\kappa_2}\\
  0 & 1
 \end{pmatrix}.
 $$
 It then follows from Proposition~\ref{prop:calculating_Gamma} that
 $$\Gamma_{13}(\kappa)=\left\{
 \begin{pmatrix}
  \frac{1}{\kappa_2}\\
  0\\
  \frac{\kappa_1}{\kappa_2}  
 \end{pmatrix}+a
 \begin{pmatrix}
  1 \\ 1 \\ 0
 \end{pmatrix}\,:\, a\in\RR
 \right\},$$
 which in turn implies that
 \begin{equation}\label{eq:gamma_toymodel}
\hat\Gamma_{13}(\kappa)=\left\{
 \begin{pmatrix}
  \frac{1}{\kappa_2}\\
  0  
 \end{pmatrix}+a
 \begin{pmatrix}
  1 \\ 1
 \end{pmatrix}\,:\, a\in\RR
 \right\}  
 \end{equation}
 and together with Theorem~\ref{thm:lcif} that the ACR value of $A$ is $\kappa_1/\kappa_2$. Note that in this case
 \begin{equation}\label{eq:st_toymodel}
\St=\spann_{\RR}\begin{pmatrix}
                    1 \\ -1
                   \end{pmatrix}\quad\text{and}\quad
      \Cons=\spann_{\RR}\begin{pmatrix}
                    1 \\ 1
                   \end{pmatrix}.  
 \end{equation}
Moreover, each connected component contains exactly one terminal component. Hence, Proposition~\ref{prop:structure_Gamma} applies, and as a matter of fact
$$\hat\Gamma_{13}(\kappa)=
 \begin{pmatrix}
  \frac{1}{\kappa_2}\\
  0  
 \end{pmatrix}+\Cons.
$$
}
Finally, while it is clear from \eqref{eq:gamma_toymodel} that a vector with zero first component and a vector with zero second component are in $\hat\Gamma_{13}(\kappa)$, this could be derived from Corollary~\ref{cor:structural} and from \eqref{eq:st_toymodel} without explicitly calculating $\hat\Gamma_{13}(\kappa)$. As a consequence, due to Theorem~\ref{thm:control}, disturbances could be applied to the production and degradation rates of $A$ (or $B$) while maintaining the absolute concentration robustness of the species $A$, its ACR value, and the linear CIF described in Theorem~\ref{thm:lcif}. 

\subsection{EnvZ-OmpR osmoregulatory system}
Consider the mass-action system
$$
\schemestart
EnvZ-D  \arrow(1--2){<=>[*0{$\kappa_1$}][*0{$\kappa_2[\ADP]$}]} EnvZ \arrow(@2--3){<=>[*0{$\kappa_3[\ATP]$}][*0{$\kappa_4$}]} EnvZ-T \arrow(@3--4){->[*0{$\kappa_5$}]} EnvZ-P
\arrow(@1.south east--.north east){0}[-90,.50]
EnvZ-P \+ OmpR \arrow(5--6){<=>[*0{$\kappa_6$}][*0{$\kappa_7$}]} EnvZ-OmpR-P \arrow(@6--7){->[*0{$\kappa_8$}]} EnvZ \+ OmpR-P
\arrow(@5.south east--.north east){0}[-90,.50]
EnvZ-D \+ OmpR-P \arrow(8--9){<=>[*0{$\kappa_9$}][*0{$\kappa_{10}$}]} EnvZ-OmpR-D-P \arrow(@9--10){->[*0{$\kappa_{11}$}]} EnvZ-D \+ OmpR
 \schemestop
$$
Order both the species and the complexes according to their appearance order, from left to right and from top to bottom. Hence, the 8 species are ordered as EnvZ-D, EnvZ, EnvZ-T, EnvZ-P, OmpR, EnvZ-OmpR-P, OmpR-P, and EnvZ-OmpR-D-P. The 10 complexes are ordered as EnvZ-D, EnvZ, EnvZ-T, EnvZ-P, EnvZ-P+OmpR, EnvZ-OmpR-P, EnvZ+OmpR-P, EnvZ-D+OmpR-P, EnvZ-OmpR-D-P, and EnvZ-D+OmpR. 

It can be checked that $\dim \St=6$ and
\begin{equation}\label{eq:cons_envz}
\Cons=\spann_\RR\left\{
\begin{pmatrix}
 1 \\ 1 \\ 1\\ 1\\ 0 \\ 1 \\ 0\\ 1
\end{pmatrix},
\begin{pmatrix}
 0 \\ 0 \\ 0\\ 0\\ 1 \\ 1\\ 1\\ 1
\end{pmatrix}
\right\}.
\end{equation}
The above conservation laws correspond to the fact that the total mass of the chemical species containing some form of EnvZ is conserved, as well as the total mass of the chemical species containing some form of OmpR. We can calculate the deficiency as
$$\delta=m-\ell-\dim\St=10-3-6=1.$$

The non-terminal complexes that only differ for the entry relative to OmpR-P are the first and the eighth ones. Therefore, we are interested in the study of $\hat\Gamma_{18}(\kappa)$, which by Theorem~\ref{thm:feinberg_implies_cond} is not empty. We have
$$
Y=\begin{pmatrix}
   1 & 0 & 0 & 0 & 0 & 0 & 0 & 1 & 0 & 1\\
0 & 1 & 0 & 0 & 0 & 0 & 1 & 0 & 0 & 0\\
0 & 0 & 1 & 0 & 0 & 0 & 0 & 0 & 0 & 0\\
0 & 0 & 0 & 1 & 1 & 0 & 0 & 0 & 0 & 0\\
0 & 0 & 0 & 0 & 1 & 0 & 0 & 0 & 0 & 1\\
0 & 0 & 0 & 0 & 0 & 1 & 0 & 0 & 0 & 0\\
0 & 0 & 0 & 0 & 0 & 0 & 1 & 1 & 0 & 0\\
0 & 0 & 0 & 0 & 0 & 0 & 0 & 0 & 1 & 0
  \end{pmatrix}$$
and
$$A(\kappa)=\begin{pmatrix} 
-\kappa_{1} & \kappa_2\mathrm{[ADP]} & 0 & 0 & 0 & 0 & 0 & 0 & 0 & 0\\ \kappa_{1} & -\kappa_2\mathrm{[ADP]}-\kappa_3\mathrm{[ATP]} & \kappa_{4} & 0 & 0 & 0 & 0 & 0 & 0 & 0\\ 0 & \kappa_3\mathrm{[ATP]} & -\kappa_{4}-\kappa_{5} & 0 & 0 & 0 & 0 & 0 & 0 & 0\\ 0 & 0 & \kappa_{5} & 0 & 0 & 0 & 0 & 0 & 0 & 0\\ 0 & 0 & 0 & 0 & -\kappa_{6} & \kappa_{7} & 0 & 0 & 0 & 0\\ 0 & 0 & 0 & 0 & \kappa_{6} & -\kappa_{7}-\kappa_{8} & 0 & 0 & 0 & 0\\ 0 & 0 & 0 & 0 & 0 & \kappa_{8} & 0 & 0 & 0 & 0\\ 0 & 0 & 0 & 0 & 0 & 0 & 0 & -\kappa_{9} & \kappa_{10} & 0\\ 0 & 0 & 0 & 0 & 0 & 0 & 0 & \kappa_{9} & -\kappa_{10}-\kappa_{11} & 0\\ 0 & 0 & 0 & 0 & 0 & 0 & 0 & 0 & \kappa_{11} & 0
\end{pmatrix}.
$$
In order to define the corresponding symbolic matrices \ml{Y} and \ml{Ak} in Matlab, we can first define the symbolic positive real variables 
\begin{lstlisting}[style=Matlab-bw]
k=sym('k', [1,11], 'positive');
k(2)=k(2)*sym('ADP', 'positive');
k(3)=k(3)*sym('ATP', 'positive');
\end{lstlisting}
Then we calculate $\Psi_{18}(\kappa)$, as defined in \eqref{eq:Psi}, via
\begin{lstlisting}[style=Matlab-bw]
e=eye(10);
simplify(null([Ak'*Y' e(:,8) e(:,1)]))
\end{lstlisting}
The output of the last command is the following matrix, whose columns are a basis of $\Psi_{18}(\kappa)$.
{\renewcommand{\arraystretch}{2}
$$
   \begin{pmatrix}
   -1 & 1 & -\frac{\mathrm{[ADP]}\,\kappa_{2}\,\kappa_{11}\,\left(\kappa_{4}+\kappa_{5}\right)}{\mathrm{[ATP]}\,\kappa_{1}\,\kappa_{3}\,\kappa_{5}\,\left(\kappa_{10}+\kappa_{11}\right)}\\ -1 & 1 & -\frac{\mathrm{[ADP]}\,\kappa_{2}\,\kappa_{4}\,\kappa_{11}+\mathrm{[ADP]}\,\kappa_{2}\,\kappa_{5}\,\kappa_{11}+\mathrm{[ATP]}\,\kappa_{3}\,\kappa_{5}\,\kappa_{10}+\mathrm{[ATP]}\,\kappa_{3}\,\kappa_{5}\,\kappa_{11}}{\mathrm{[ATP]}\,\kappa_{1}\,\kappa_{3}\,\kappa_{5}\,\left(\kappa_{10}+\kappa_{11}\right)}\\ -1 & 1 & -\frac{\mathrm{[ADP]}\,\kappa_{2}\,\kappa_{4}\,\kappa_{11}+\mathrm{[ADP]}\,\kappa_{2}\,\kappa_{5}\,\kappa_{10}+2\,\mathrm{[ADP]}\,\kappa_{2}\,\kappa_{5}\,\kappa_{11}+\mathrm{[ATP]}\,\kappa_{3}\,\kappa_{5}\,\kappa_{10}+\mathrm{[ATP]}\,\kappa_{3}\,\kappa_{5}\,\kappa_{11}}{\mathrm{[ATP]}\,\kappa_{1}\,\kappa_{3}\,\kappa_{5}\,\left(\kappa_{10}+\kappa_{11}\right)}\\ -1 & 1 & -\frac{\mathrm{[ADP]}\,\kappa_{2}\,\kappa_{4}\,\kappa_{10}+2\,\mathrm{[ADP]}\,\kappa_{2}\,\kappa_{4}\,\kappa_{11}+\mathrm{[ADP]}\,\kappa_{2}\,\kappa_{5}\,\kappa_{10}+2\,\mathrm{[ADP]}\,\kappa_{2}\,\kappa_{5}\,\kappa_{11}+\mathrm{[ATP]}\,\kappa_{3}\,\kappa_{5}\,\kappa_{10}+\mathrm{[ATP]}\,\kappa_{3}\,\kappa_{5}\,\kappa_{11}}{\mathrm{[ATP]}\,\kappa_{1}\,\kappa_{3}\,\kappa_{5}\,\left(\kappa_{10}+\kappa_{11}\right)}\\ 1 & 0 & \frac{\mathrm{[ADP]}\,\kappa_{2}\,\left(\kappa_{4}+\kappa_{5}\right)}{\mathrm{[ATP]}\,\kappa_{1}\,\kappa_{3}\,\kappa_{5}}\\ 0 & 1 & -\frac{\mathrm{[ADP]}\,\kappa_{2}\,\kappa_{4}\,\kappa_{11}+\mathrm{[ADP]}\,\kappa_{2}\,\kappa_{5}\,\kappa_{11}+\mathrm{[ATP]}\,\kappa_{3}\,\kappa_{5}\,\kappa_{10}+\mathrm{[ATP]}\,\kappa_{3}\,\kappa_{5}\,\kappa_{11}}{\mathrm{[ATP]}\,\kappa_{1}\,\kappa_{3}\,\kappa_{5}\,\left(\kappa_{10}+\kappa_{11}\right)}\\ 1 & 0 & 0\\ 0 & 1 & 0\\ 0 & 0 & -\frac{\mathrm{[ADP]}\,\kappa_{2}\,\kappa_{9}\,\kappa_{11}\,\left(\kappa_{4}+\kappa_{5}\right)}{\mathrm{[ATP]}\,\kappa_{1}\,\kappa_{3}\,\kappa_{5}\,\left(\kappa_{10}+\kappa_{11}\right)}\\ 0 & 0 & 1 
   \end{pmatrix}
$$
}
For convenience, denote the last vector by $\zeta(\kappa)$. It follows from Proposition~\ref{prop:calculating_Gamma} that
\begin{equation*}
 \Gamma_{18}(\kappa)=
 \left\{
  -\pi_9(\zeta(\kappa))
  +a_1
  \begin{pmatrix}
   -1 \\ -1 \\ -1\\ -1\\ 1 \\ 0\\ 1\\ 0\\ 0
  \end{pmatrix}
  +a_2
  \begin{pmatrix}
   1 \\ 1 \\ 1\\ 1\\ 0 \\ 1 \\ 0\\ 1\\ 0
  \end{pmatrix}
  \,:\, a_1,a_2\in\RR
 \right\}.
\end{equation*}
It follows that
\begin{equation}\label{eq:hat_gamma_envz}
 \hat\Gamma_{18}(\kappa)=
 \left\{
  -\pi_8(\zeta(\kappa))
  +a_1
  \begin{pmatrix}
   -1 \\ -1 \\ -1\\ -1\\ 1 \\ 0\\ 1\\ 0
  \end{pmatrix}
  +a_2
  \begin{pmatrix}
   1 \\ 1 \\ 1\\ 1\\ 0 \\ 1 \\ 0\\ 1
  \end{pmatrix}
  \,:\, a_1,a_2\in\RR
 \right\},
\end{equation}
which corresponds to \eqref{eq:calculation_ex_hat} in the main text. Moreover, it follows from Theorem~\ref{thm:lcif} that the ACR value of the species OmpR-P is
$$-\zeta_{9}(\kappa)=\frac{\mathrm{[ADP]}\,\kappa_{2}\,\kappa_{9}\,\kappa_{11}\,\left(\kappa_{4}+\kappa_{5}\right)}{\mathrm{[ATP]}\,\kappa_{1}\,\kappa_{3}\,\kappa_{5}\,\left(\kappa_{10}+\kappa_{11}\right)}$$
as reported in \eqref{eq:calculation_ex} in the main text.

Further things can be noted about this model. First, it follows from \eqref{eq:cons_envz} and \eqref{eq:hat_gamma_envz} that
$$\hat\Gamma_{18}(\kappa)=-\pi_8(\zeta(\kappa))+\Cons,$$
which is in accordance with Proposition~\ref{prop:structure_Gamma} because all connected components contain exactly one terminal component. 
Secondly, due to Corollary~\ref{cor:structural}, we can deduce without explicitly calculating $\Gamma_{18}(\kappa)$ that there is a vector $\hat\gamma$ in $\hat\Gamma_{18}(\kappa)$ whose entries relative to species EnvZ and OmpR-P are zero (which are the second and the seventh complexes, respectively). Specifically, if we consider the basis of $\Cons$ given in \eqref{eq:cons_envz} and we let $X_{l_1}=X_2=$EnvZ and $X_{l_2}=X_7=$OmpR-P, we have
$$
V=\begin{pmatrix}
   1 & 0\\
   0 & 1
  \end{pmatrix},
$$
which has rank 2. The existence of such vector $\hat\gamma$ is used in Section~\ref{sec:large} and it implies by Theorem~\ref{thm:control} that the system is robust to persistent disturbances affecting the production and degradation rates of both EnvZ and OmpR-P.

\section{An ACR signaling system covered by our theory and not by \cite{SF:ACR}}

Consider the double-phosphorylation mass-action system in Figure~\ref{fig:double}. We will show that the theory developed in \cite{SF:ACR} stays silent on whether it is ACR. However, our theory covers this case and implies the double-phosphorylated form of the transcriptional regulatory protein is ACR, for any choice of kinetic parameters $\kappa$ such that a positive steady state exists. Note that this form of response robustness may seem a bit surprising for a multisite phosphorylation mechanism, since these are often known for their multi-stability properties, notably shown in the case of the MAPK pathway \cite{AFS:detection,TTT:MAPK, FB:MAPK, FSWFS:MAPK}.

\begin{figure}[h!]
\begin{center}
 \begin{tcolorbox}[width=\textwidth, colback=olive!5!white,colframe=olive!50!white]
  Activation and phosphorylation of the sensor-transmitter protein:
  $$
  \schemestart
  $A$ \arrow(1--2){<=>[*0{$\kappa_1$[\I]}][*0{$\kappa_2$}]} $\Active$ \arrow(@2--3){->[*0{$\kappa_3$}]} $A$-\Ph
  \schemestop
  $$
 \DrawLine{olive!50!white}
  Phosphorylation of the sensory response protein:
$$
\schemestart
$A$-\Ph \+ $B$ \arrow(4--5){<=>[*0{$\kappa_4$}][*0{$\kappa_5$}]} $A$-$B$-\Ph \arrow(@5--6){->[*0{$\kappa_6$}]} $A$ \+ $B$-\Ph
\arrow(@4.south east--.north east){0}[-90,.50]
$A$-\Ph \+ $B$-\Ph \arrow(7--8){<=>[*0{$\kappa_7$}][*0{$\kappa_8$}]} $A$-$B$-\Ph\Ph \arrow(@8--9){->[*0{$\kappa_9$}]} $A$ \+ $B$-\Ph\Ph
\schemestop
$$
\DrawLine{olive!50!white}
  Activation and phosphorylation of the protein aggregate:
$$
\schemestart
$A$-$B$-\Ph \arrow(10--11){<=>[*0{$\kappa_{10}$[\I]}][$\kappa_{11}$]} $\Active$-$B$-\Ph 
\arrow(@10.south east--.north east){0}[-90,.50]
$A$-$B$-\Ph\Ph \arrow(12--13){<=>[*0{$\kappa_{12}$[\I]}][$\kappa_{13}$]} $\Active$-$B$-\Ph\Ph
\arrow(@12.south east--.north east){0}[-90,.50]
$\Active$-$B$-\Ph \arrow(14--15){->[*0{$\kappa_{14}$}]} $\Active$-$B$-\Ph\Ph \arrow(@15--16){->[*0{$\kappa_{15}$}]} $A$-\Ph \+ $B$-\Ph\Ph
\schemestop
$$
\DrawLine{olive!50!white}
  Dephosphorylation of the sensory response protein:
$$
\schemestart
$A$ \+ $B$-\Ph\Ph \arrow(17--18){<=>[*0{$\kappa_{16}$}][*0{$\kappa_{17}$}]} $B$-\Ph\Ph-$A$ \arrow(@18--19){->[*0{$\kappa_{18}$}]} $A$-\Ph \+ $B$
 \schemestop
$$
 \end{tcolorbox}
 \caption{A model of signal transduction. Here, the sensory response protein $B$ is active if two phosphoryl groups are attached. In the first line, the protein $A$ is phosphorylated in response to a stimulus [\I]. In the second and third line, the phosphoryl groups are transferred to the sensory response protein $B$. In the forth and fifth line, the protein $A$ responds by the stimulus [I] while it is bound to $B$. The resulting protein aggregate is able to gain phosphoryl groups until the three phosphoryl sites are occupied (as described in the sixth line). Finally, in the last two lines, the inactive form of the sensor-transmitter protein $A$ acts as a phosphatase on $B$.}
 \label{fig:double}
 \end{center}
\end{figure}
 
 Order the species and the complexes according to their appearance order, from left to right and from top to bottom. The 11 species are then ordered as $A$, $A^\star$, $A$-P, $B$, $A$-$B$-P, $B$-P, $A$-$B$-PP, $B$-PP, $A^\star$-$B$-P, $A^\star$-$B$-PP, and $B$-PP-$A$. The 13 complexes are ordered as $A$, $A^\star$, $A$-P, $A$-P + $B$, $A$-$B$-P, $A$ + $B$-P, $A$-P + $B$-P, $A$-$B$-PP, $A$ + $B$-PP, $A^\star$-$B$-P, $A^\star$-$B$-PP, $A$-P + $B$-PP, and $B$-PP-$A$.
It can be checked that $\dim \St=9$ and
\begin{equation}\label{eq:cons_double}
\Cons=\spann_\RR\left\{
\begin{pmatrix}
 1\\ 1\\ 1\\ 0\\ 1\\ 0\\ 1\\ 0\\ 1\\ 1\\ 1
\end{pmatrix},
\begin{pmatrix}
 0\\ 0\\ 0\\ 1\\ 1\\ 1\\ 1\\ 1\\ 1\\ 1\\ 1
\end{pmatrix}
\right\}.
\end{equation}
Similarly as for the EnvZ-OmpR osmoregulatory system, the above conservation laws correspond express that the total mass of the chemical species containing some form of the protein $A$ is conserved, as well as the total mass of the chemical species containing some form of the protein $B$. The reaction graph is given in Figure~\ref{fig:double_graph}. In particular, $\ell=2$ and the deficiency is
$$\delta=m-\ell-\dim\St=13-2-9=2.$$
So, this model is not included in the class of models studied in \cite{SF:ACR}, and Theorems~\ref{thm:feinberg_weak} and \ref{thm:feinberg_strong} do not apply.
 
 \begin{figure}[h!]
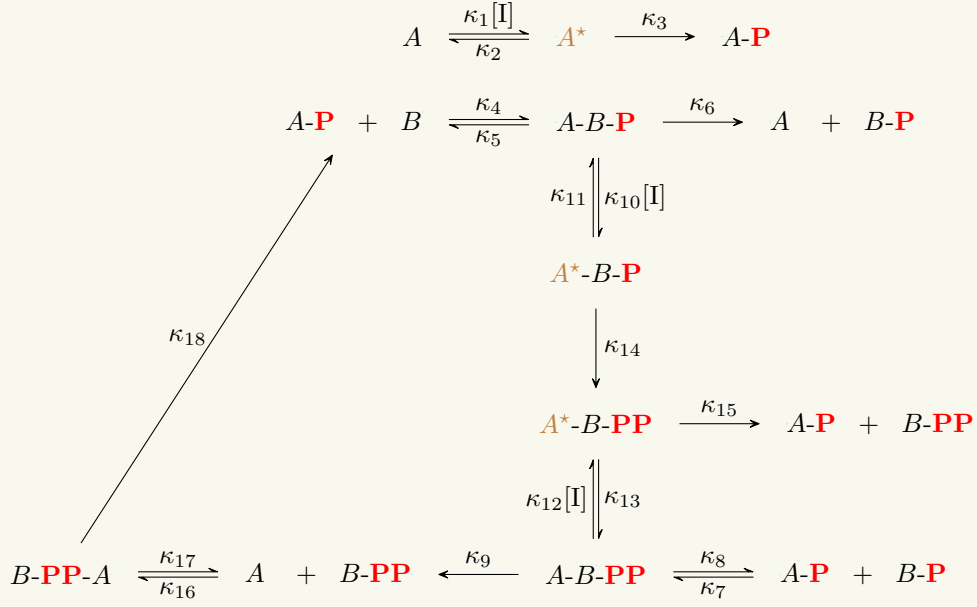

\begin{center}
 \begin{tcolorbox}[width=\textwidth, colback=olive!5!white,colframe=olive!50!white]
 $$
  \schemestart
   $A$ \arrow(1--2){<=>[*0{$\kappa_1$[\I]}][*0{$\kappa_2$}]} $\Active$ \arrow(@2--3){->[*0{$\kappa_3$}]} $A$-\Ph
   \arrow(@1.south east--.north east){0}[-90,.50]
   $A$-\Ph \+ $B$ \arrow(4--5){<=>[*0{$\kappa_4$}][*0{$\kappa_5$}]} $A$-$B$-\Ph \arrow(@5--6){->[*0{$\kappa_6$}]} $A$ \+ $B$-\Ph
   \arrow(@5--10){<=>[*0{$\kappa_{10}$[\I]}][*0{$\kappa_{11}$}]}[-90] $\Active$-$B$-\Ph 
   \arrow(@10--11){->[*0{$\kappa_{14}$}]}[-90] $\Active$-$B$-\Ph\Ph
   \arrow(@11--8){<=>[*0{$\kappa_{13}$}][*0{$\kappa_{12}$[\I]}]}[-90]$A$-$B$-\Ph\Ph 
   \arrow(@11--12){->[*0{$\kappa_{15}$}]} $A$-\Ph \+ $B$-\Ph\Ph
   \arrow(@8--9){0}[180] $A$ \+ $B$-\Ph\Ph
   \arrow(@9--13){0}[180] $B$-\Ph\Ph-$A$
   \arrow(@8--7){<=>[*0{$\kappa_8$}][*0{$\kappa_7$}]} $A$-\Ph \+ $B$-\Ph 
   \arrow(@9--@8){<-[*0{$\kappa_9$}]} 
   \arrow(@13--@9){<=>[*0{$\kappa_{17}$}][*0{$\kappa_{16}$}]} 
   \arrow(@13--@4){->[*0{$\kappa_{18}$}]}
   \schemestop
  $$
 \end{tcolorbox}
 \caption{Reaction graph for the signaling transduction system proposed in Figure~\ref{fig:double}.}
 \label{fig:double_graph}
 \end{center}
\end{figure}

The first and the ninth complexes are non-terminal, and they only differ for the entry relative B-PP. Hence, in order to apply Theorems~\ref{thm:lcif} and \ref{thm:control}, we are interested in studying $\hat\Gamma_{19}(\kappa)$. Since the deficiency of the model is 2, Theorem~\ref{thm:feinberg_implies_cond} does not apply, so to understand whether $\hat\Gamma_{19}(\kappa)$ is non-empty we need to explicitly calculate it. To this aim, we define in Matlab the following positive symbolic variables, which correspond to the rate constants of the model:
\begin{lstlisting}[style=Matlab-bw]
k=sym('k', [1,18], 'positive');
k(1)=k(1)*sym('I', 'positive');
k(10)=k(10)*sym('I', 'positive');
k(12)=k(12)*sym('I', 'positive');
\end{lstlisting}
We then define the symbolic matrices \ml{Y} and \ml{Ak} corresponding to
$$Y=\left(\begin{array}{ccccccccccccc}
1 & 0 & 0 & 0 & 0 & 1 & 0 & 0 & 1 & 0 & 0 & 0 & 0\\ 0 & 1 & 0 & 0 & 0 & 0 & 0 & 0 & 0 & 0 & 0 & 0 & 0\\ 0 & 0 & 1 & 1 & 0 & 0 & 1 & 0 & 0 & 0 & 0 & 1 & 0\\ 0 & 0 & 0 & 1 & 0 & 0 & 0 & 0 & 0 & 0 & 0 & 0 & 0\\ 0 & 0 & 0 & 0 & 1 & 0 & 0 & 0 & 0 & 0 & 0 & 0 & 0\\ 0 & 0 & 0 & 0 & 0 & 1 & 1 & 0 & 0 & 0 & 0 & 0 & 0\\ 0 & 0 & 0 & 0 & 0 & 0 & 0 & 1 & 0 & 0 & 0 & 0 & 0\\ 0 & 0 & 0 & 0 & 0 & 0 & 0 & 0 & 1 & 0 & 0 & 1 & 0\\ 0 & 0 & 0 & 0 & 0 & 0 & 0 & 0 & 0 & 1 & 0 & 0 & 0\\ 0 & 0 & 0 & 0 & 0 & 0 & 0 & 0 & 0 & 0 & 1 & 0 & 0\\ 0 & 0 & 0 & 0 & 0 & 0 & 0 & 0 & 0 & 0 & 0 & 0 & 1\end{array}\right)$$
and to $A(\kappa)$ as described in Figure~\ref{fig:Ak}, respectively.
\begin{sidewaysfigure}[h!]
$$A(\kappa)=\left(\begin{array}{ccccccccccccc}
-\kappa_{1}\mathrm{[I]} & \kappa_{2} & 0 & 0 & 0 & 0 & 0 & 0 & 0 & 0 & 0 & 0 & 0\\
\kappa_{1}\mathrm{[I]} & -\kappa_{2}-\kappa_{3} & 0 & 0 & 0 & 0 & 0 & 0 & 0 & 0 & 0 & 0 & 0\\
0 & \kappa_{3} & 0 & 0 & 0 & 0 & 0 & 0 & 0 & 0 & 0 & 0 & 0\\ 0 & 0 & 0 & -\kappa_{4} & \kappa_{5} & 0 & 0 & 0 & 0 & 0 & 0 & 0 & \kappa_{18}\\
0 & 0 & 0 & \kappa_{4} & -\kappa_{5}-\kappa_{6}-\kappa_{10}\mathrm{[I]} & 0 & 0 & 0 & 0 & \kappa_{11} & 0 & 0 & 0\\
0 & 0 & 0 & 0 & \kappa_{6} & 0 & 0 & 0 & 0 & 0 & 0 & 0 & 0\\
0 & 0 & 0 & 0 & 0 & 0 & -\kappa_{7} & \kappa_{8} & 0 & 0 & 0 & 0 & 0\\
0 & 0 & 0 & 0 & 0 & 0 & \kappa_{7} & -\kappa_{8}-\kappa_{9}-\kappa_{12}\mathrm{[I]} & 0 & 0 & \kappa_{13} & 0 & 0\\
0 & 0 & 0 & 0 & 0 & 0 & 0 & \kappa_{9} & -\kappa_{16} & 0 & 0 & 0 & \kappa_{17}\\
0 & 0 & 0 & 0 & \kappa_{10}\mathrm{[I]} & 0 & 0 & 0 & 0 & -\kappa_{11}-\kappa_{14} & 0 & 0 & 0\\
0 & 0 & 0 & 0 & 0 & 0 & 0 & \kappa_{12}\mathrm{[I]} & 0 & \kappa_{14} & -\kappa_{13}-\kappa_{15} & 0 & 0\\
0 & 0 & 0 & 0 & 0 & 0 & 0 & 0 & 0 & 0 & \kappa_{15} & 0 & 0\\
0 & 0 & 0 & 0 & 0 & 0 & 0 & 0 & \kappa_{16} & 0 & 0 & 0 & -\kappa_{17}-\kappa_{18}\end{array}\right)$$
\caption{$A(k)$ for the mass-action system in Figure~\ref{fig:double}.}
\label{fig:Ak}
\end{sidewaysfigure}

In order to calculate a basis for $\Psi_{19}(\kappa)$ as defined in \eqref{eq:Psi}, we type
\begin{lstlisting}[style=Matlab-bw]
e=eye(13);
simplify(null([Ak'*Y' e(:,1) e(:,9)]))
\end{lstlisting}
The output of the last command is shown in Figure~\ref{fig:psi}. While the output is rather complicated, we did not need to put much effort in its calculation, which was completed in Matlab in a matter of seconds. For convenience we denote by $\zeta(\kappa)$ the last column of the output matrix. 
\begin{sidewaysfigure}[h!]
{\renewcommand{\arraystretch}{1.5}
$$\begin{pmatrix} -1 & 1 & \frac{1}{\kappa_{16}}\\ -1 & 1 & \frac{p_1(\kappa,\mathrm{[I]})}{\kappa_{16}\kappa_{18}\left(\kappa_{2}+\kappa_{3}\right)\left(-\kappa_{10}\kappa_{12}\kappa_{14}\kappa_{15}\mathrm{[I]}^2-\kappa_{9}\kappa_{10}\kappa_{14}\kappa_{15}\mathrm{[I]}+\kappa_{6}\kappa_{9}\kappa_{11}\kappa_{13}+\kappa_{6}\kappa_{9}\kappa_{11}\kappa_{15}+\kappa_{6}\kappa_{9}\kappa_{13}\kappa_{14}+\kappa_{6}\kappa_{9}\kappa_{14}\kappa_{15}\right)}\\ -1 & 1 & \frac{p_2(\kappa, \mathrm{[I]})}{\kappa_{16}\kappa_{18}\left(-\kappa_{10}\kappa_{12}\kappa_{14}\kappa_{15}\mathrm{[I]}^2-\kappa_{9}\kappa_{10}\kappa_{14}\kappa_{15}\mathrm{[I]}+\kappa_{6}\kappa_{9}\kappa_{11}\kappa_{13}+\kappa_{6}\kappa_{9}\kappa_{11}\kappa_{15}+\kappa_{6}\kappa_{9}\kappa_{13}\kappa_{14}+\kappa_{6}\kappa_{9}\kappa_{14}\kappa_{15}\right)}\\ 1 & 0 & -\frac{\left(\kappa_{17}+\kappa_{18}\right)\left(2\kappa_{6}\kappa_{9}\kappa_{11}\kappa_{13}+2\kappa_{6}\kappa_{9}\kappa_{11}\kappa_{15}+2\kappa_{6}\kappa_{9}\kappa_{13}\kappa_{14}+2\kappa_{6}\kappa_{9}\kappa_{14}\kappa_{15}+\mathrm{[I]}\kappa_{6}\kappa_{11}\kappa_{12}\kappa_{15}+\mathrm{[I]}\kappa_{9}\kappa_{10}\kappa_{13}\kappa_{14}+\mathrm{[I]}\kappa_{6}\kappa_{12}\kappa_{14}\kappa_{15}\right)}{\kappa_{16}\kappa_{18}\left(-\kappa_{10}\kappa_{12}\kappa_{14}\kappa_{15}\mathrm{[I]}^2-\kappa_{9}\kappa_{10}\kappa_{14}\kappa_{15}\mathrm{[I]}+\kappa_{6}\kappa_{9}\kappa_{11}\kappa_{13}+\kappa_{6}\kappa_{9}\kappa_{11}\kappa_{15}+\kappa_{6}\kappa_{9}\kappa_{13}\kappa_{14}+\kappa_{6}\kappa_{9}\kappa_{14}\kappa_{15}\right)}\\ 0 & 1 & -\frac{\kappa_{17}}{\kappa_{16}\kappa_{18}}\\ 1 & 0 & -\frac{\kappa_{9}\left(\kappa_{13}+\kappa_{15}\right)\left(\kappa_{17}+\kappa_{18}\right)\left(\kappa_{6}\kappa_{11}+\kappa_{6}\kappa_{14}+\mathrm{[I]}\kappa_{10}\kappa_{14}\right)}{\kappa_{16}\kappa_{18}\left(-\kappa_{10}\kappa_{12}\kappa_{14}\kappa_{15}\mathrm{[I]}^2-\kappa_{9}\kappa_{10}\kappa_{14}\kappa_{15}\mathrm{[I]}+\kappa_{6}\kappa_{9}\kappa_{11}\kappa_{13}+\kappa_{6}\kappa_{9}\kappa_{11}\kappa_{15}+\kappa_{6}\kappa_{9}\kappa_{13}\kappa_{14}+\kappa_{6}\kappa_{9}\kappa_{14}\kappa_{15}\right)}\\ 0 & 1 & \frac{p_3(\kappa,\mathrm{[I]})}{\kappa_{16}\kappa_{18}\left(-\kappa_{10}\kappa_{12}\kappa_{14}\kappa_{15}\mathrm{[I]}^2-\kappa_{9}\kappa_{10}\kappa_{14}\kappa_{15}\mathrm{[I]}+\kappa_{6}\kappa_{9}\kappa_{11}\kappa_{13}+\kappa_{6}\kappa_{9}\kappa_{11}\kappa_{15}+\kappa_{6}\kappa_{9}\kappa_{13}\kappa_{14}+\kappa_{6}\kappa_{9}\kappa_{14}\kappa_{15}\right)}\\ 1 & 0 & 0\\ 0 & 1 & \frac{\kappa_{6}\kappa_{9}\kappa_{13}\kappa_{14}\kappa_{18}-\kappa_{6}\kappa_{9}\kappa_{11}\kappa_{15}\kappa_{17}-\kappa_{6}\kappa_{9}\kappa_{11}\kappa_{13}\kappa_{17}+\kappa_{6}\kappa_{9}\kappa_{14}\kappa_{15}\kappa_{17}+2\kappa_{6}\kappa_{9}\kappa_{14}\kappa_{15}\kappa_{18}+\mathrm{[I]}^2\kappa_{10}\kappa_{12}\kappa_{14}\kappa_{15}\kappa_{17}+\mathrm{[I]}\kappa_{6}\kappa_{12}\kappa_{14}\kappa_{15}\kappa_{17}+\mathrm{[I]}\kappa_{6}\kappa_{12}\kappa_{14}\kappa_{15}\kappa_{18}+\mathrm{[I]}\kappa_{9}\kappa_{10}\kappa_{14}\kappa_{15}\kappa_{17}}{\kappa_{16}\kappa_{18}\left(-\kappa_{10}\kappa_{12}\kappa_{14}\kappa_{15}\mathrm{[I]}^2-\kappa_{9}\kappa_{10}\kappa_{14}\kappa_{15}\mathrm{[I]}+\kappa_{6}\kappa_{9}\kappa_{11}\kappa_{13}+\kappa_{6}\kappa_{9}\kappa_{11}\kappa_{15}+\kappa_{6}\kappa_{9}\kappa_{13}\kappa_{14}+\kappa_{6}\kappa_{9}\kappa_{14}\kappa_{15}\right)}\\ 0 & 1 & \frac{p_4(\kappa, \mathrm{[I]})}{\kappa_{16}\kappa_{18}\left(-\kappa_{10}\kappa_{12}\kappa_{14}\kappa_{15}\mathrm{[I]}^2-\kappa_{9}\kappa_{10}\kappa_{14}\kappa_{15}\mathrm{[I]}+\kappa_{6}\kappa_{9}\kappa_{11}\kappa_{13}+\kappa_{6}\kappa_{9}\kappa_{11}\kappa_{15}+\kappa_{6}\kappa_{9}\kappa_{13}\kappa_{14}+\kappa_{6}\kappa_{9}\kappa_{14}\kappa_{15}\right)}\\ 0 & 1 & 0\\ 0 & 0 & -\frac{\mathrm{[I]}\kappa_{1}\kappa_{3}\left(\kappa_{17}+\kappa_{18}\right)\left(\kappa_{6}\kappa_{11}+\kappa_{6}\kappa_{14}+\mathrm{[I]}\kappa_{10}\kappa_{14}\right)\left(\kappa_{9}\kappa_{13}+\kappa_{9}\kappa_{15}+\mathrm{[I]}\kappa_{12}\kappa_{15}\right)}{\kappa_{16}\kappa_{18}\left(\kappa_{2}+\kappa_{3}\right)\left(-\kappa_{10}\kappa_{12}\kappa_{14}\kappa_{15}\mathrm{[I]}^2-\kappa_{9}\kappa_{10}\kappa_{14}\kappa_{15}\mathrm{[I]}+\kappa_{6}\kappa_{9}\kappa_{11}\kappa_{13}+\kappa_{6}\kappa_{9}\kappa_{11}\kappa_{15}+\kappa_{6}\kappa_{9}\kappa_{13}\kappa_{14}+\kappa_{6}\kappa_{9}\kappa_{14}\kappa_{15}\right)}\\ 0 & 0 & 1
  \end{pmatrix}
$$
}
where
\begin{align*}
 p_1(\kappa, \mathrm{[I]})&=\kappa_{2}\kappa_{6}\kappa_{9}\kappa_{11}\kappa_{13}\kappa_{18}+\kappa_{3}\kappa_{6}\kappa_{9}\kappa_{11}\kappa_{13}\kappa_{17}+2\kappa_{3}\kappa_{6}\kappa_{9}\kappa_{11}\kappa_{13}\kappa_{18}+\kappa_{2}\kappa_{6}\kappa_{9}\kappa_{11}\kappa_{15}\kappa_{18}+\kappa_{3}\kappa_{6}\kappa_{9}\kappa_{11}\kappa_{15}\kappa_{17}+\kappa_{2}\kappa_{6}\kappa_{9}\kappa_{13}\kappa_{14}\kappa_{18}+2\kappa_{3}\kappa_{6}\kappa_{9}\kappa_{11}\kappa_{15}\kappa_{18}\\
 &\quad+\kappa_{3}\kappa_{6}\kappa_{9}\kappa_{13}\kappa_{14}\kappa_{17}+2\kappa_{3}\kappa_{6}\kappa_{9}\kappa_{13}\kappa_{14}\kappa_{18}+\kappa_{2}\kappa_{6}\kappa_{9}\kappa_{14}\kappa_{15}\kappa_{18}+\kappa_{3}\kappa_{6}\kappa_{9}\kappa_{14}\kappa_{15}\kappa_{17}+2\kappa_{3}\kappa_{6}\kappa_{9}\kappa_{14}\kappa_{15}\kappa_{18}+\mathrm{[I]}\kappa_{3}\kappa_{6}\kappa_{11}\kappa_{12}\kappa_{15}\kappa_{17}\\
 &\quad+\mathrm{[I]}\kappa_{3}\kappa_{6}\kappa_{11}\kappa_{12}\kappa_{15}\kappa_{18}+\mathrm{[I]}\kappa_{3}\kappa_{9}\kappa_{10}\kappa_{13}\kappa_{14}\kappa_{17}+\mathrm{[I]}\kappa_{3}\kappa_{6}\kappa_{12}\kappa_{14}\kappa_{15}\kappa_{17}+\mathrm{[I]}\kappa_{3}\kappa_{9}\kappa_{10}\kappa_{13}\kappa_{14}\kappa_{18}-\mathrm{[I]}\kappa_{2}\kappa_{9}\kappa_{10}\kappa_{14}\kappa_{15}\kappa_{18}+\mathrm{[I]}\kappa_{3}\kappa_{6}\kappa_{12}\kappa_{14}\kappa_{15}\kappa_{18}\\
 &\quad+\mathrm{[I]}\kappa_{3}\kappa_{9}\kappa_{10}\kappa_{14}\kappa_{15}\kappa_{17}-\mathrm{[I]}^2\kappa_{2}\kappa_{10}\kappa_{12}\kappa_{14}\kappa_{15}\kappa_{18}+\mathrm{[I]}^2\kappa_{3}\kappa_{10}\kappa_{12}\kappa_{14}\kappa_{15}\kappa_{17}\\
 p_2(\kappa, \mathrm{[I]})&=\kappa_{6}\kappa_{9}\kappa_{11}\kappa_{13}\kappa_{17}+2\kappa_{6}\kappa_{9}\kappa_{11}\kappa_{13}\kappa_{18}+\kappa_{6}\kappa_{9}\kappa_{11}\kappa_{15}\kappa_{17}+2\kappa_{6}\kappa_{9}\kappa_{11}\kappa_{15}\kappa_{18}+\kappa_{6}\kappa_{9}\kappa_{13}\kappa_{14}\kappa_{17}+2\kappa_{6}\kappa_{9}\kappa_{13}\kappa_{14}\kappa_{18}+\kappa_{6}\kappa_{9}\kappa_{14}\kappa_{15}\kappa_{17}\\
 &\quad+2\kappa_{6}\kappa_{9}\kappa_{14}\kappa_{15}\kappa_{18}+\mathrm{[I]}^2\kappa_{10}\kappa_{12}\kappa_{14}\kappa_{15}\kappa_{17}+\mathrm{[I]}\kappa_{6}\kappa_{11}\kappa_{12}\kappa_{15}\kappa_{17}+\mathrm{[I]}\kappa_{6}\kappa_{11}\kappa_{12}\kappa_{15}\kappa_{18}+\mathrm{[I]}\kappa_{9}\kappa_{10}\kappa_{13}\kappa_{14}\kappa_{17}+\mathrm{[I]}\kappa_{6}\kappa_{12}\kappa_{14}\kappa_{15}\kappa_{17}\\
 &\quad+\mathrm{[I]}\kappa_{9}\kappa_{10}\kappa_{13}\kappa_{14}\kappa_{18}+\mathrm{[I]}\kappa_{6}\kappa_{12}\kappa_{14}\kappa_{15}\kappa_{18}+\mathrm{[I]}\kappa_{9}\kappa_{10}\kappa_{14}\kappa_{15}\kappa_{17}\\
 p_3(\kappa,\mathrm{[I]})&=\kappa_{6}\kappa_{9}\kappa_{11}\kappa_{13}\kappa_{18}+\kappa_{6}\kappa_{9}\kappa_{11}\kappa_{15}\kappa_{18}+\kappa_{6}\kappa_{9}\kappa_{13}\kappa_{14}\kappa_{18}+\kappa_{6}\kappa_{9}\kappa_{14}\kappa_{15}\kappa_{18}+\mathrm{[I]}^2\kappa_{10}\kappa_{12}\kappa_{14}\kappa_{15}\kappa_{17}+\mathrm{[I]}\kappa_{6}\kappa_{11}\kappa_{12}\kappa_{15}\kappa_{17}+\mathrm{[I]}\kappa_{6}\kappa_{11}\kappa_{12}\kappa_{15}\kappa_{18}\\
 &\quad+\mathrm{[I]}\kappa_{6}\kappa_{12}\kappa_{14}\kappa_{15}\kappa_{17}+\mathrm{[I]}\kappa_{6}\kappa_{12}\kappa_{14}\kappa_{15}\kappa_{18}-\mathrm{[I]}\kappa_{9}\kappa_{10}\kappa_{14}\kappa_{15}\kappa_{18}\\
 p_4(\kappa,\mathrm{[I]})&=\kappa_{6}\kappa_{9}\kappa_{11}\kappa_{13}\kappa_{18}+\kappa_{6}\kappa_{9}\kappa_{11}\kappa_{15}\kappa_{17}+2\kappa_{6}\kappa_{9}\kappa_{11}\kappa_{15}\kappa_{18}+\kappa_{6}\kappa_{9}\kappa_{13}\kappa_{14}\kappa_{18}+\kappa_{6}\kappa_{9}\kappa_{14}\kappa_{15}\kappa_{17}+2\kappa_{6}\kappa_{9}\kappa_{14}\kappa_{15}\kappa_{18}+\mathrm{[I]}^2\kappa_{10}\kappa_{12}\kappa_{14}\kappa_{15}\kappa_{17}\\
 &\quad+\mathrm{[I]}\kappa_{6}\kappa_{11}\kappa_{12}\kappa_{15}\kappa_{17}+\mathrm{[I]}\kappa_{6}\kappa_{11}\kappa_{12}\kappa_{15}\kappa_{18}+\mathrm{[I]}\kappa_{6}\kappa_{12}\kappa_{14}\kappa_{15}\kappa_{17}+\mathrm{[I]}\kappa_{6}\kappa_{12}\kappa_{14}\kappa_{15}\kappa_{18}+\mathrm{[I]}\kappa_{9}\kappa_{10}\kappa_{14}\kappa_{15}\kappa_{17}
\end{align*}
\caption{$\ker \Psi_{19}(\kappa)$ for the mass-action system in Figure~\ref{fig:double}.}
\label{fig:psi}
\end{sidewaysfigure}

It follows from Proposition~\ref{prop:calculating_Gamma} that $\Gamma_{19}(\kappa)$ is non-empty, and is given by
\begin{equation}\label{eq:gamma_double}
\Gamma_{19}(\kappa)=\left\{-\pi_{12}(\zeta(\kappa))
  +a_1
  \begin{pmatrix}
   -1 \\ -1 \\ -1\\ 1\\ 0 \\ 1\\ 0\\ 1\\ 0\\ 0\\ 0\\ 0
  \end{pmatrix}
  +a_2
  \begin{pmatrix}
   1 \\ 1 \\ 1\\ 0\\ 1\\ 0\\ 1\\ 0\\ 1\\ 1\\ 1\\ 0
  \end{pmatrix}
  \,:\, a_1,a_2\in\RR
\right\}.
\end{equation}
It follows from \eqref{eq:gamma_double} and \eqref{eq:cons_double} that
$$
\hat\Gamma_{19}(\kappa)=\left\{-\pi_{11}(\zeta(\kappa))
  +a_1
  \begin{pmatrix}
   -1 \\ -1 \\ -1\\ 1\\ 0 \\ 1\\ 0\\ 1\\ 0\\ 0\\ 0
  \end{pmatrix}
  +a_2
  \begin{pmatrix}
   1 \\ 1 \\ 1\\ 0\\ 1\\ 0\\ 1\\ 0\\ 1\\ 1\\ 1
  \end{pmatrix}
  \,:\, a_1,a_2\in\RR
\right\}=-\pi_{11}(\zeta(\kappa))+\Cons,
$$
which is in accordance with Proposition~\ref{prop:structure_Gamma} because every connected component of the reaction graph in Figure~\ref{fig:double_graph} contains exactly one terminal component. Moreover, from \eqref{eq:gamma_double} and Theorem~\ref{thm:lcif} it follows that the ACR value of the ACR species $B$-PP is
$$-\zeta_{12}(\kappa)=\frac{\kappa_1\kappa_3[\I](\kappa_{17}+\kappa_{18})(\kappa_6\kappa_{11}+\kappa_6\kappa_{14}+\kappa_{10}\kappa_{14}[\I])(\kappa_9\kappa_{13}+\kappa_9\kappa_{15}+\kappa_{12}\kappa_{15}[\I])}{\kappa_{16}\kappa_{18}(\kappa_2+\kappa_3)g(\kappa, [\I])},$$
 where
 $$g(\kappa, [\I])=-\kappa_{10}\kappa_{12}\kappa_{14}\kappa_{15}[\I]^2-\kappa_{9}\kappa_{10}\kappa_{14}\kappa_{15}[\I]+\kappa_{6}\kappa_{9}\kappa_{11}\kappa_{13}+\kappa_{6}\kappa_{9}\kappa_{11}\kappa_{15}+\kappa_{6}\kappa_{9}\kappa_{13}\kappa_{14}+\kappa_{6}\kappa_{9}\kappa_{14}\kappa_{15}.$$
This means that if a positive steady state $c$ exists, necessarily it entry relative to $B$-PP (which is the eighth species) satisfies $c_8=-\zeta_{12}(\kappa)$. Of course, this cannot occur if $-\zeta_{12}(\kappa)$ is non-positive, i.e.\ if $g(\kappa, [\I])$ is non-positive, in which case no positive steady state exists. Note that we did not need to work directly with the differential equation to derive this information. We will further show that in fact a positive steady state exists if and only if $g(\kappa, [\I])$ is positive. To this aim, note that $c$ is a steady state if and only if $\Lambda(c)\in\ker YA(\kappa)$. Then, we calculate a basis for $\ker YA(\kappa)$ by typing
\begin{lstlisting}[style=Matlab-bw]
simplify(null(Y*Ak))
\end{lstlisting}
which returns a matrix of the form
$$
\begin{pmatrix}
 0 & 0 & 0 & b_1(\kappa, [\I])g(\kappa, [\I])\\
 0 & 0 & 0 & b_2(\kappa, [\I])g(\kappa, [\I])\\
 1 & 0 & 0 & 0\\
 0 & 0 & 0 & b_4(\kappa, [\I])\\
 0 & 0 & 0 & b_5(\kappa, [\I])\\
 0 & 1 & 0 & 0\\
 0 & 0 & 0 & b_7(\kappa, [\I])\\
 0 & 0 & 0 & b_8(\kappa, [\I])\\
 0 & 0 & 0 & b_9(\kappa, [\I])\\
 0 & 0 & 0 & b_{10}(\kappa, [\I])\\
 0 & 0 & 0 & b_{11}(\kappa, [\I])\\
 0 & 0 & 1 & 0\\
 0 & 0 & 0 & 1
\end{pmatrix}
$$
for some functions $b_i(\kappa, [\I])$ which map positive arguments to positive real numbers. Hence, a positive steady state exists if and only if there exists a positive vector $c$ such that
\begin{align*}
 c_1&=a_4 b_1(\kappa, [\I])g(\kappa, [\I]) & c_1c_6&=a_2 & c_9&=a_4 b_{10}(\kappa, [\I])\\
 c_2&=a_4 b_2(\kappa, [\I])g(\kappa, [\I]) & c_3c_6&=a_4 b_7(\kappa, [\I]) & c_{10}&=a_4 b_{11}(\kappa, [\I])\\
 c_3&=a_1 & c_7&=a_4 b_8(\kappa, [\I]) & c_3c_8&=a_3\\
 c_3c_4&=a_4 b_4(\kappa, [\I]) & c_1c_8&=a_4 b_9(\kappa, [\I]) &  c_{11}&=a_4\\
 c_5&=a_4 b_5(\kappa, [\I])
\end{align*}
for some $a_1, a_2, a_3, a_4\in\RR_{>0}$. If $g(\kappa, [\I])$ is positive, it is easy to see that such a positive vector $c$ exists. Specifically, the positive steady states are parameterized by
\begin{align*}
 c_1&=a_4 b_1(\kappa, [\I])g(\kappa, [\I]) & c_5&=a_4 b_5(\kappa, [\I]) & c_9&=a_4 b_{10}(\kappa, [\I])\\
 c_2&=a_4 b_2(\kappa, [\I])g(\kappa, [\I]) & c_6&=\frac{a_4}{a_1}b_7(\kappa, [\I]) & c_{10}&=a_4 b_{11}(\kappa, [\I])\\
 c_3&=a_1 & c_7&=a_4 b_8(\kappa, [\I]) & c_{11}&=a_4\\
 c_4&=\frac{a_4}{a_1} b_4(\kappa, [\I]) & c_8&=\frac{b_9(\kappa, [\I])}{b_1(\kappa, [\I])g(\kappa, [\I])}
\end{align*}
where $a_1$ and $a_4$ vary in $\RR_{>0}$. Note that the entry $c_8$ is relative to the ACR species $B$-PP and is fixed. 

\end{document}